\documentclass[lettersize,journal]{IEEEtran}
\usepackage{amsmath,amsfonts,amssymb}
\newtheorem{theorem}{Theorem}
\newtheorem{proof}{Proof}
\usepackage{graphicx}
\usepackage{subfigure}
\usepackage[colorlinks]{hyperref}
\usepackage{multirow,multicol}
\usepackage{makecell}

\begin{document}
\title{Fluid Antenna System-Assisted\\Self-Interference Cancellation for In-Band\\Full Duplex Communications}
\author{Hanjiang Hong,~\IEEEmembership{Member,~IEEE}, 
        Kai-Kit Wong,~\IEEEmembership{Fellow,~IEEE}, 
        Hao Xu,~\IEEEmembership{Senior Member,~IEEE},\\
        Yiyan Wu,~\IEEEmembership{Life Fellow,~IEEE},
        Sai Xu,~\IEEEmembership{Member,~IEEE},
        Chan-Byoung Chae, \emph{Fellow, IEEE},\\
        Baiyang Liu,~\IEEEmembership{Senior Member,~IEEE}, and 
        Kin-Fai Tong, \emph{Fellow, IEEE}
\vspace{-8mm}

\thanks{The work of K. K. Wong is supported by the Engineering and Physical Sciences Research Council (EPSRC) under Grant EP/W026813/1.}
\thanks{The work of H. Hong is supported by the Outstanding Doctoral Graduates Development Scholarship of Shanghai Jiao Tong University.}
\thanks{The work of C.-B. Chae was in part supported by the Institute for Information and Communication Technology Planning and Evaluation (IITP)/NRF grant funded by the Ministry of Science and ICT (MSIT), South Korea, under Grant RS-2024-00428780 and 2022R1A5A1027646.}
\thanks{The work of K. F. Tong and B. Liu was funded by the Hong Kong Metropolitan University, Staff Research Startup Fund: FRSF/2024/03.}

\thanks{H. Hong, K. K. Wong, and S. Xu are with the Department of Electronic and Electrical Engineering, University College London, London, United Kingdom. K. K. Wong is also with Yonsei Frontier Lab, Yonsei University, Seoul, South Korea (e-mail: $\rm \{hanjiang.hong, kai\text{-}kit.wong, sai.xu\}@ucl.ac.uk$).}
\thanks{H. Xu is with the National Mobile Communications Research Laboratory, Southeast University, Nanjing 210096, China (e-mail: $\rm hao.xu@seu.edu.cn$).}
\thanks{Y. Wu is affiliated with the Department of Electrical and Computer Engineering, Western University, London, ON N6A 3K7, Canada (e-mail: $\rm yiyan.wu@ieee.org$).}
\thanks{C.-B. Chae is with the School of Integrated Technology, Yonsei University, Seoul, 03722 South Korea (e-mail: $\rm cbchae@yonsei.ac.kr$).}
\thanks{B. Liu and K. F. Tong are with the School of Science and Technology, Hong Kong Metropolitan University, Hong Kong SAR, China (e-mail: $\rm \{byliu, ktong\}@hkmu.edu.hk$).}

\thanks{Corresponding author: Kai-Kit Wong.}
}
\maketitle
\begin{abstract}
In-band full-duplex (IBFD) systems are expected to double the spectral efficiency compared to half-duplex systems, provided that loopback self-interference (SI) can be effectively suppressed. The inherent interference mitigation capabilities of the emerging fluid antenna system (FAS) technology make it a promising candidate for addressing the SI challenge in IBFD systems. This paper thus proposes a FAS-assisted self-interference cancellation (SIC) framework, which leverages a receiver-side FAS to dynamically select an interference-free port. Analytical results include a lower bound and an approximation of the residual SI (RSI) power, both derived for rich-scattering channels by considering the joint spatial correlation amongst the FAS ports. Simulations of RSI power and forward link rates validate the analysis, showing that the SIC performance improves with the number of FAS ports. Additionally, simulations under practical conditions, such as finite-scattering environments and wideband integrated access and backhaul (IAB) channels, reveal that the proposed approach offers superior SIC capability and significant forward rate gains over conventional IBFD SIC schemes.

\end{abstract}


\section{Introduction}
\IEEEPARstart{T}{he advent of} six-generation (6G) wireless networks is presenting an unprecedented challenge---a mobile technology that needs to achieve a collection of ambitious performance indicators such as terabits per second peak rate, sub-millisecond latency, and ultra-dense connectivity, etc \cite{Ericsson6G,wang2023ontheroad,Tariq-2020}, and is anticipated to provide a range of integrated services \cite{xu2024intelligent}. This is difficult because spectrum governs how much data we can send reliably over a wireless channel but it is precious. To address this, we either find more bandwidth by moving up the frequency band or attempt to utilize the existing bandwidth more efficiently. A prominent technology of the latter is full-duplex communication that allows the transmitted and received data to share the same physical channel \cite{Sabharwal-2010}.

In-band full duplex (IBFD) communication has gained much attention in recent years, for its superior spectral efficiency by permitting simultaneous transmission and reception within the same frequency band \cite{sabharwal2014inband,Kim2015asurvey,Kim2024asurvey,mohammadi2023acomprehensive}. In contrast, traditional half-duplex systems require two separate channels to separate uplink and downlink communications. Although more standardization efforts are needed, the 3rd Generation Partnership Project (3GPP) and numerous industry endeavors continue to advance towards IBFD, achieving successes, such as integrated access and backhaul (IAB) from Release 16 to Release 18  \cite{suk2022full,chen20235Gadvanced} and subband full duplex in Release 19 \cite{Abdelghaffar2024subband}. 

Despite its theoretical appeal, self-interference (SI)---where the loopback signal (LBS) from a co-located transmitter overwhelms the receiver---continues to be a significant barrier to the realization of IBFD systems. Self-interference cancellation (SIC) techniques have been developed to tackle the SI problem \cite{mohammadi2023acomprehensive}. Various SIC techniques have been formulated and can be broadly divided into two main categories: passive suppression and active cancellation. Passive suppression approaches \cite{everett2014passive} alleviate SI in the propagation domain before it is processed by the receiver circuitry, while active cancellation techniques \cite{mohammadi2023acomprehensive} mitigate SI through the reconstruction and subtraction of the SI from the received signal. Active cancellation techniques can be divided into analog and digital SIC techniques based on the signal domain where the SI is subtracted.

Typical IBFD systems usually deploy both passive suppression and active cancellation techniques to achieve significant SI mitigation. Active cancellation techniques can handle high-power SI based on training sequence or adaptive interference cancellation techniques. Within this framework, LBS serves as a reference signal for SI. By utilizing channel state information (CSI) obtained at the receiver, filter weights can be derived and applied to the reference signal for the reconstruction of the SI-inverse signal. This resultant SI-inverse signal is subsequently combined with the received signal to form an SI null. In \cite{bha2013full}, a training-based SIC scheme was proposed, which uses the training phase to derive adaptive filter weights to reconstruct the SI-inverse signal. Later in \cite{mas2017channel}, a method was introduced for estimating the channel at baseband during the training phase, and then using the estimated CSI in RF SIC. Recently, \cite{le2022atwo} presented a two-stage analog filtering structure to cancel both the direct leakage SI and the residual self-interference (RSI). Also, the frequency-domain radio frequency (RF) SIC (FD-RF-SIC) in \cite{Hong2023frequency} optimized the filter weights based on discrete Fourier transform (DFT)-windowing, effectively cancelling the SI in frequency domain. Machine learning approaches for SIC have also recently been developed  \cite{muranov2021ondeep, dong2024augmentation, elsayed2025ahybrid}.

Going forward, techniques capable of enhancing the active SIC performance are always welcome. One such technology is the fluid antenna system (FAS) \cite{wong2021FAS,wong2020FAS,New2024aTutorial,Lu-2025} which can give IBFD the needed additional degree-of-freedom (DoF) for an effective active SIC. FAS is a groundbreaking technology that facilitates dynamic and real-time reconfiguration of antenna positions, structures, and configurations to empower the physical layer of wireless communications, encouraged by recent advances in reconfigurable antennas, e.g., \cite{Hoang-2021,Deng-2023}, with FAS prototypes reported in \cite{Shen-tap_submit2024,zhang2024pixel,Liu-2025arxiv}. 

The concept of FAS was first introduced in 2020 \cite{wong2020FAS}, and has since been widely studied in different channel models \cite{Khammassi2023,new2024fluid,new2023information}. Recent attempts have also focused on the channel estimation problem for FAS \cite{xu2023channel,new2025channel,zhang2025successive} and its integration into fifth-generation (5G) new radio (NR) systems \cite{hong2024coded,hong2025Downlink,hong2025fluid}. Furthermore, FAS exhibits substantial potential for multiuser communications by finding the port for interference null \cite{wong2022FAMA,Wong2024cuma,Xu2024revisiting,xu2023capacity}. Additionally, there is a new branch of research integrating the idea of FAS into reconfigurable intelligent surface (RIS) technologies \cite{ghadi2024on,salem2025first,xiao2025fluid,ghadi2025fires}.

The inherent flexibility in antenna positioning characteristic of FAS renders it particularly well-suited for SIC techniques. Hence, the synergy between FAS and IBFD makes sense and is important for combating SI by leveraging the additional DoF provided by FAS. In 2023, Skouroumounis and Krikidis applied FAS in a full duplex network \cite{sko2023full}. They analyzed the performance achieved by large-scale FAS-assisted full duplex networks and the effects of channel estimation on the overall network performance. However, their approach identified the port with the strongest forward signal (FWS) and treated SI as a normal interference, neglecting the use of CSI from LBS for effective SI mitigation. This oversight is likely to result in performance degradation within the IBFD system. 

Motivated by the above, this paper proposes a FAS-assisted SIC approach. This strategy incorporates a FAS component at the receiver, which effectively observes the fading envelopes across the available space. Consequently, the fluid antenna can be realigned to positions where the desired FWS is strong and LBS is in a deep fade. This tactical adjustment of FAS mitigates SI substantially by selecting the optimal port, hence minimizing the RSI power. Moreover, this approach functions solely on the principles of FAS and has the potential to be integrated with active SIC techniques, enhancing the utilization of the reference signal for SI-inverse signal reconstruction and cancellation. System validation will be elucidated through simulations, as detailed in Section \ref{subsec:combination}.

Our main contributions are summarized as follows:
\begin{itemize}
\item A FAS-assisted SIC framework is introduced, exploiting the spatial DoF of the fluid antenna at the receiver. The proposed approach utilizes the receiver's FAS to identify the port that experiences minimal SI. The framework is initially developed under the assumption of rich scattering channels, and subsequently extended to finite-scattering channels and wideband channels.
\item The performance of the FAS-assisted SIC approach is analyzed in rich-scattering channels. A closed-form lower bound on the average RSI power is first derived. This bound represents a special case in the absence of spatial correlation among the FAS ports. Secondly, we approximate the average RSI power as a $2M$-fold surface integral, taking into consideration $M$ dominant eigenvalues within the eigenvalue-based channel model in \cite{Khammassi2023}.
\item Simulation results validate the performance of both the lower bound and the proposed approximation. The findings demonstrate that the SIC capability enhances with the number of FAS ports, $N$, given the fixed normalized size of FAS, $W$, in an ideal scenario with perfect CSI. However, an increase in $N$ necessitates a larger number of pilot-training symbols, offsetting the capacity gain. Also, for any given $N$, the SIC capability initially improves before approaching the derived bound. This indicates that increasing the FAS size can significantly elevate the SIC ability, particularly when $W$ is small. 
\item Besides, the simulations are also carried out under more realistic finite-scattering channels \cite{buzzi2016on} and the wideband IAB channel \cite{suk2022full}. The results reveal that the SIC capability shown in rich scattering environments serves as the upper bound for that observed in finite-scattering and IAB channels. But the overall trends of the results remain consistent. Finally, we integrate the proposed FAS-assisted SIC with the FD-RF-SIC in \cite{Hong2023frequency}, and the results confirm the compatibility of our proposed approach.
\end{itemize}

The rest of this paper is organized as follows. Section~\ref{sec:SystemModel} presents the proposed FAS-assisted SIC framework. Section~\ref{sec:AnalyticalResults} analyzes the performance by deriving the lower bound and an approximation of the RSI power. Simulation results are shown in Section~\ref{sec:Sim}. Finally, Section~\ref{sec:conclusion} concludes this paper.

{\em Notations:} Scalars are represented by lowercase letters while vectors and matrices are denoted by lowercase and uppercase boldface letters, respectively. Also, transpose and hermitian operations are denoted by superscript $T$ and $\dag$, respectively. In addition, for a complex scalar $x$, $\lvert x \rvert$ and $x^\dag$ represent its modulus and conjugate, respectively.

\begin{figure}
\centering
\subfigure[]{\includegraphics[width = 0.8\linewidth]{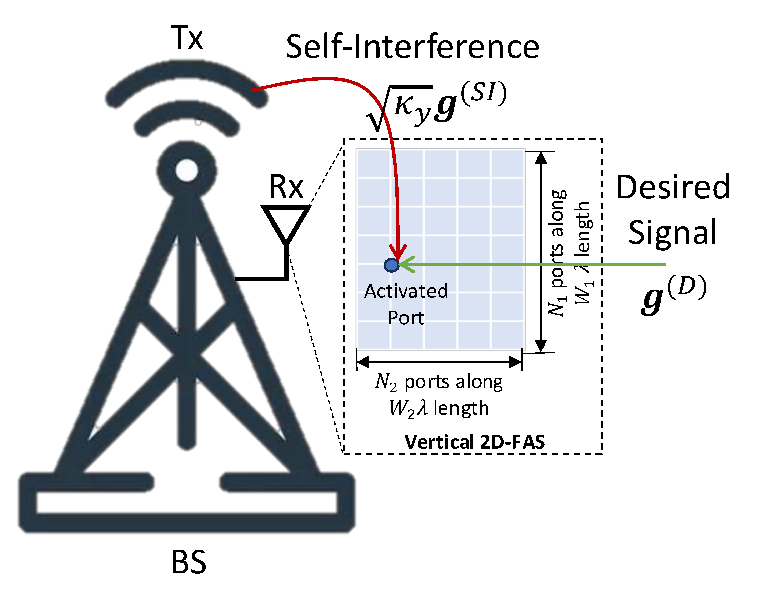}}
\subfigure[]{\includegraphics[width = 0.8\linewidth]{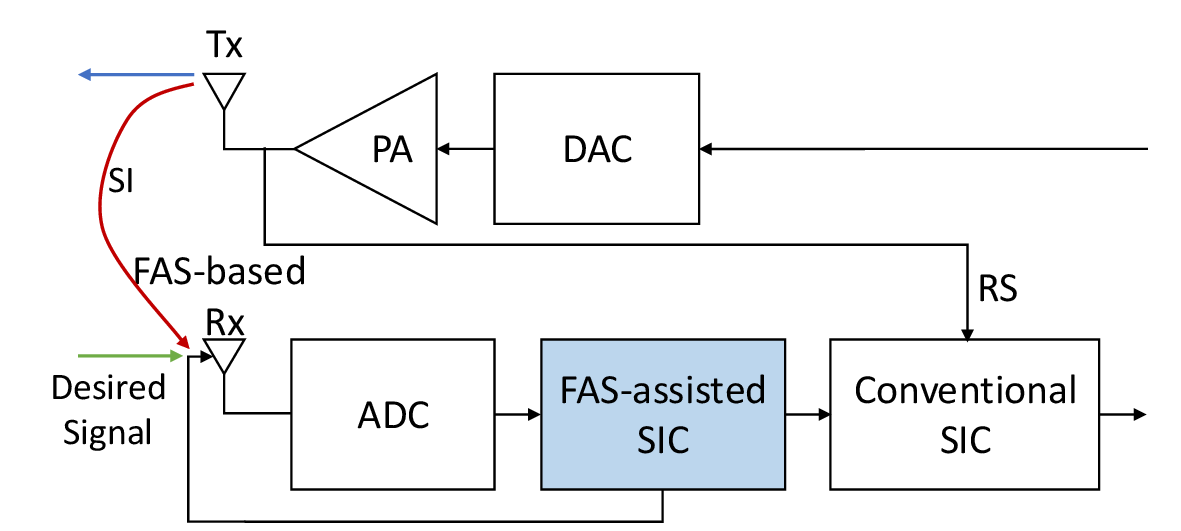}}
\caption{Illustration of an IBFD system with a FAS-assisted receiver.}\label{Fig:SysModel}
\end{figure}

\section{FAS-assisted SIC Framework}\label{sec:SystemModel}
As illustrated in Fig.~\ref{Fig:SysModel}, we consider a single-input single-output (SISO) IBFD communication system wherein the transmitter (Tx) is equipped with a fixed-position antenna (FPA) while the receiver (Rx) features a two-dimensional FAS (2D-FAS). The fluid antenna within the 2D-FAS configuration can be instantaneously switched to one of the $N = (N_1 \times N_2)$ ports, which are uniformly distributed across a 2D area of $W = W_1 \lambda \times W_2 \lambda$, where $\lambda$ denotes the carrier wavelength. In this IBFD system, the receiver simultaneously receives the desired FWS from the remote transmitter and the LBS from the co-located transmitter antenna. The fundamental objective of SIC is to subtract the loopback SI from the received signal, thereby enabling the following processing of the interference-free desired signal through demodulation and decoding. In this paper, we propose to utilize FAS for SIC, termed FAS-assisted SIC. The idea of FAS-assisted SIC is to utilize the FAS at the receiver to identify the port with the lowest RSI power, thus naturally mitigating the SI and enhancing the FWS reception. Notably, the proposed FAS-assisted SIC approach can function independently or in tandem with other conventional active SIC techniques. Here, we focus on the FAS-assisted SIC.

\subsection{System Model}
For notational simplicity, we define the mapping of the port as $(n_1, n_2) \to n: n = n_1\times N_2 +n_2$, where $n_1 \in \{0, \dots, N_1 -1\}$, $n_2 \in \{0, \dots, N_2 -1\}$, and $n \in \{0, \dots, N-1\}$. At the FAS-assisted receiver, the received signal at the $n$-th port with time-index omitted is given by
\begin{equation}\label{Eq:r_narr}
    y_n = g^{({\rm D})}_n s_{\rm D} +\sqrt{\kappa_y} g^{({\rm SI})}_n \Psi(s_{\rm SI})+\eta_{n},
\end{equation}
in which $s_{\rm D} \sim \mathcal{CN}(0,E_s)$ is the desired FWS at the receiver, $g^{({\rm D})}_n \sim \mathcal{CN}(0,\sigma_g^2)$ is the forward channel coefficient from the remote transmit antenna to the receiver, $s_{\rm SI} \sim \mathcal{CN}(0,E_s)$ denotes the loopback SI signal, $\Psi(\cdot)$ is the nonlinear distortion function associated with the SI, $g^{({\rm SI})}_n \sim \mathcal{CN}(0,\sigma_g^2)$ represents the loopback channel coefficient from the co-located transmitter to the receiver, $\kappa_Y$ is the average power ratio between $g^{({\rm SI})}_n$ and $g^{({\rm D})}_n$, and $\eta_{n}\sim \mathcal{CN}(0,\sigma_\eta^2)$ is the additive white Gaussian noise. Here, $E_s$ and $\sigma_\eta^2$ denote the power of signal and noise, respectively, and $\sigma_g^2$ is the variance of the channel.

It can be established from \eqref{Eq:r_narr} that an estimate of the FWS $s_D$ at the $n$-th port can be formulated as
\begin{align}
\hat{s}_{{\rm D},n} &= \frac{\left(g^{({\rm D})}_n \right)^\dag}{|g^{({\rm D})}_n|^2}y_n\notag\\ 
&= \! s_{\rm D} \!+\! \sqrt{\kappa_y} \frac{\! \left(g^{({\rm D})}_n \right)^\dag\!\!}{|g^{({\rm D})}_n|^2} g^{({\rm SI})}_n \Psi(s_{{\rm SI}})+\frac{\! \left(g^{({\rm D})}_n \right)^\dag\!\!}{|g^{({\rm D})}_n|^2} \eta_{n}, 
\end{align}
where the second component represents the RSI. The power of RSI at the $n$-th port, denoted as $P_n$, is expressed as
\begin{align}
P_n & \hspace{.5mm}= \kappa_y \frac{|g_n^{({\rm SI})}|^2}{|g_n^{({\rm D})}|^2} \mathbb{E}\left\{\left|\Psi(s_{{\rm SI}})\right|^2\right\}\notag\\
& \overset{(a)}{=} \kappa_y E_{\rm SI} \frac{|g_n^{({\rm SI})}|^2}{|g_n^{({\rm D})}|^2},\label{Eq:PRSI_n}
\end{align}
where $(a)$ is obtained from the assumption of constant power of nonlinear SI, i.e., $ \mathbb{E}\left\{\left|\Psi(s_{{\rm SI}})\right|^2\right\} = E_{\rm SI}$. Thus, the signal-to-interference-plus-noise ratio (SINR) observed at the $n$-th port for the desired FWS, denoted as $\gamma^{({\rm D})}_n$, is written as
\begin{align}
\gamma^{({\rm D})}_n &\hspace{.5mm} = \frac{E_s |g_n^{({\rm D})}|^2}{\kappa_y E_{\rm SI} {|g_n^{({\rm SI})}|^2} + {\sigma_\eta^2}}\notag\\ 
& \overset{(a)}{\approx} \frac{E_s |g_n^{({\rm D})}|^2}{\kappa_y E_{\rm SI} {|g_n^{({\rm SI})}|^2}} = \frac{E_s}{P_{n}},\label{Eq:sir_n}
\end{align}
where $(a)$ assumes that the noise power is considerably less than the interference power, thereby allowing the SINR to be approximated by the signal-to-interference ratio (SIR).

\subsection{Channel Model}
In the narrowband channel model, the channel coefficients $\boldsymbol{g}^{({\rm Z})} = [g^{({\rm Z})}_0,\dots,g^{({\rm Z})}_{N-1}]^T,~{\rm Z}\in \{{\rm D},{\rm SI}\}$ exhibit fluctuations but should remain constant during the transmission block period. Let $\sigma_g^2 \boldsymbol{\Sigma}$ be the covariance matrix of channel $\boldsymbol{g}^{({\rm Z})}$, i.e., $\mathbb{E}\left[ \boldsymbol{g}^{({\rm Z})} (\boldsymbol{g}^{({\rm Z})})^ \dag \right] = \sigma_g^2 \boldsymbol{\Sigma}$.  Following the framework outlined in \cite{Khammassi2023}, the spatial correlation is characterized according to Jake's model when the channel undergoes rich scattering. As such, the $(n,m)$-th entry of $\boldsymbol{\Sigma}$ is given by
\begin{equation}\label{Eq:corr}
\left[\boldsymbol{\Sigma}\right]_{n,m} \!\!=\! J_0 \! \left(\! 2\pi \sqrt{\! \left(\frac{n_1 \!-\! m_1}{N_1 \!-\! 1} W_1 \!\right)^2 \!\!+\! \left(\frac{n_2 \!-\! m_2}{N_2 - 1} W_2 \! \right)^2}\right),
\end{equation}
where $(m_1, m_2) \to m$, and $J_0(\cdot)$ is the zero-order Bessel function of the first order. Typically, by utilizing the eigenvalue-based model from \cite{Khammassi2023}, $g^{({\rm Z})}_n$ can be expressed as
\begin{equation}\label{Eq:chan_narr}
    g^{({\rm Z})}_n = \sigma_g \sum_{m=0}^{N-1} \sqrt{\lambda_m}\mu_{n,m}a^{({\rm Z})}_m,
\end{equation}
or in a vector form as 
\begin{equation}\label{Eq:chan_narr_vec}
    \boldsymbol{g}^{({\rm Z})} = \sigma_g \boldsymbol{U} \boldsymbol{\Lambda}^{\frac{1}{2}} \boldsymbol{a}^{({\rm Z})},
\end{equation}
in which $\boldsymbol{a}^{({\rm Z})} = [a^{({\rm Z})}_0,\dots,a^{({\rm Z})}_{N-1}]^T$, with $a^{({\rm Z})}_n \sim \mathcal{CN} (0,1)$. The matrices $\boldsymbol{\Lambda}$ and $\boldsymbol{U}$ stem from the singular value decomposition (SVD) of $\boldsymbol{\Sigma}$, given by $\boldsymbol{\Sigma} = \boldsymbol{U}\boldsymbol{\Lambda}\boldsymbol{U}^\dag$, where we have $\boldsymbol{\Lambda} = {\rm diag} ( \lambda_0, \dots, \lambda_{N-1})$ and $[\boldsymbol{U}]_{n,m} = \mu_{n,m}$. 

Under realistic IBFD settings, there is a line-of-sight (LOS) between the co-located Tx and Rx at the base station (BS). Thus, we also consider the finite-scattering channel model \cite{buzzi2016on} for the loopback channel $\boldsymbol{g}^{({\rm SI})}$. This model incorporates a LOS component with Rice factor $K$ and $N_p$ scattered components. In this case, $g_n^{({\rm SI})}$ can be written as
\begin{multline}\label{Eq:FSChann}
g_n^{({\rm SI})}=\sqrt{\!\!\frac{K \sigma_g^2}{K\!+1\!}}e^{j\alpha^{({\rm SI})}}e^{-j 2\pi \left[\frac{n_1\!W_1}{N_1\!-\!1}\sin \theta_0 \cos \phi_0 + \frac{n_2\!W_2}{N_2\!-\!1}\cos \theta_0\right]}\\
+ \sum_{l=1}^{N_p} a_l^{({\rm SI})}e^{-j2\pi \left[\frac{n_1W_1}{N_1-1} \sin \theta_l \cos \phi_l+ \frac{n_2W_2}{N_2-1}\cos\theta_l\right]},
\end{multline}
where $\theta_l$ and $\phi_l$ are the elevation and azimuth angles of arrival (AOA), respectively, for $l = 0,\dots,N_p$. Here, $\alpha^{({\rm SI})}$ denotes the random phase of the LOS component, and $a_l^{({\rm SI})}$ represents the random complex coefficient of the $l$-th path. Additionally, the complex gain satisfies $\mathbb{E}\{\sum_l |a_l^{({\rm SI})}|^2\} = \sigma_g^2/(K+1)$.

\subsection{FAS-assisted SIC}
Under the assumption of perfect knowledge of $\boldsymbol{g}^{({\rm D})}$ and $\boldsymbol{g}^{({\rm SI})}$, FAS at the receiver side selects the port exhibiting the minimum of RSI power $\{P_{n}\}$ for FAS-assisted SIC. This port selection can be expressed mathematically as
\begin{equation}\label{Eq:PortSelection}
    n^* = \mathop{\arg\min}_n~ \frac{|g_n^{({\rm SI})}|^2}{|g_n^{({\rm D})}|^2}. 
\end{equation}
As such, we are interested in the random variable defined as
\begin{equation}\label{Eq:R_min}
    R = \min\left\{\frac{|g_0^{({\rm SI})}|^2}{|g_0^{({\rm D})}|^2},\cdots,\frac{|g_{N-1}^{({\rm SI})}|^2}{|g_{N-1}^{({\rm D})}|^2}\right\}.
\end{equation}
The RSI power of this FAS-assisted SIC approach is given by 
\begin{equation}\label{Eq:PRSI_Narr}
    P = \kappa_y E_{{\rm SI}} R.
\end{equation}
Subsequently, the capacity of the desired FWS can be derived from \eqref{Eq:sir_n} and \eqref{Eq:PRSI_Narr} as 
\begin{equation} \label{Eq:achievableRate_Narr}
    C = \log_2(1+\gamma^{({\rm D})})  = \log_2 \left(1+ \frac{E_s}{\kappa_y E_{{\rm SI}} R}\right).
\end{equation}

In addition to the capacity in \eqref{Eq:achievableRate_Narr}, in this paper, we also seek to analyze another key indicator, the average RSI power. This metric is defined as 
\begin{equation}\label{Eq:aveP_def}
    \overline{P}\triangleq \mathbb{E} \left\{P\right\} = \kappa_y E_{{\rm SI}}  \mathbb{E}\{R\}.
\end{equation}
By comparing the average RSI power before and after SIC, one can assess the effectiveness of the overall implementation.

\subsection{Extension to Wideband Channels}
For wideband channels, signals are decomposed into multiple components across $F$ subcarriers. The incorporation of cyclic prefix (CP) effectively mitigates inter-symbol interference (ISI) and facilitates simplified frequency-domain processing. With time-index omitted, the wideband FAS channel at the $n$-th port can be modelled in the frequency domain as
\begin{multline}\label{Eq:r_wide}
\left[\!\!\! \begin{array}{c} y_n[0] \\ y_n[1] \\ \vdots \\ y_n[F\!\!-\!\!1] \end{array} \!\!\!\right]  = \boldsymbol{G}_n^{({\rm D})} \times \! \left[\!\!\!\begin{array}{c} s_{\rm D}[0] \\ s_{\rm D}[1] \\ \vdots \\ s_{\rm D}[F\!\!-\!\!1] \end{array}\!\!\!\right]\\
+ \boldsymbol{G}_n^{({\rm SI})}\times \! \left[\!\!\!\begin{array}{c} s_{\rm SI}[0] \\ s_{\rm SI}[1] \\ \vdots \\ s_{\rm SI}[F\!\!-\!\!1] \end{array}\!\!\!\right] + \left[\!\!\! \begin{array}{c} \eta_n[0] \\ \eta_n[1] \\ \vdots \\ \eta_n[F-1] \end{array}\!\!\!\right],
\end{multline}
where the channel matrices $\boldsymbol{G}_n^{({\rm Z})} = {\rm diag} (g_n^{({\rm Z})}[0],\dots,g_n^{({\rm Z})}[F-1])$, for ${\rm Z} \in \{{\rm D},{\rm SI}\}$. The channel is subject to frequency-selective fading when the coefficients $\{g_n^{({\rm Z})}[f]\}$ vary over $f$. Considering the FAS characteristics, the channel coefficients $\{g_n^{({\rm Z})}[f]\}_{\forall n}$ are correlated among the FAS ports, and the covariance matrix is expressed as $\mathbb{E}\{\boldsymbol{g}^{({\rm Z})}[f](\boldsymbol{g}^{({\rm Z})}[f])^\dag\} = \sigma_g^2 \boldsymbol{\Sigma}$. Here, $\boldsymbol{g}^{({\rm Z})}[f] = [g_0^{({\rm Z})}[f],\dots,g_{N-1}^{({\rm Z})}[f]]^T$, and the elements within $\boldsymbol{\Sigma}$ are specified in \eqref{Eq:corr}. Notably, in the scenario where $\boldsymbol{g}^{({\rm Z})} [f] = \boldsymbol{g}^{({\rm Z})}, \forall f \in [0,F)$, indicating that the channel is subject to flat fading, the wideband channel model in \eqref{Eq:r_wide} simplifies to the channel model presented in \eqref{Eq:r_narr}.

The desired FWS of the $f$-th subcarrier at the $n$-th port can be estimated by $w_n[f] = (g_n^{({\rm D})}[f])^\dag/|g_n^{({\rm D})}[f]|^2$ as
\begin{align}
\hat{s}_{{\rm D},n}[f] &= w_n[f]y_n[f]\notag\\
&= s_{\rm D}[f]\!+\!\!\sqrt{\kappa_y}w_n[f]g_n^{({\rm SI})}\![f]\Psi(s_{\rm SI}[f]) \!\!+\! w_n[f]\eta_n[f].
\end{align}
The power of RSI at the $n$-th port in the wideband channel, denoted as $P_{n}^{{\rm W}}$, can be calculated as
\begin{align}
P_{n}^{{\rm W}} &\hspace{.5mm} = \mathbb{E}_f\left\{\left|\sqrt{\kappa_y}w_n[f]g_n^{({\rm SI})}\![f]\Psi(s_{\rm SI}[f])\right|^2\right\}\notag\\ 
& \hspace{.5mm}= \kappa_y \mathbb{E}_f\left\{\left|w_n[f]g_n^{({\rm SI})}[f]\right|^2\right\} \mathbb{E}_f\left\{\left|\Psi(s_{\rm SI}[f])\right|^2\right\}\notag\\
& \overset{(a)}{=}\frac{\kappa_y E_{\rm SI}}{F}\sum_{f=0}^{F-1}\frac{|g_n^{({\rm SI})}[f]|^2}{|g_n^{({\rm D})}[f]|^2},\label{Eq:PRSI_W}
\end{align}
where $(a)$ is obtained from the assumption of constant nonlinear SI power, i.e., $\mathbb{E}\left\{\left|\Psi(s_{{\rm SI}}[f])\right|^2\right\} = E_{{\rm SI}}$. The implementation of FAS-assisted SIC within wideband channels identifies the port $n^*_{\rm W}$ that minimizes $\{P^{{\rm W}}_{n}\}$, given by
\begin{equation} \label{Eq:PortSelection_wide}
n^*_{\rm W} = \mathop{\arg\min}_n~ \frac{1}{F} \sum_{f=0}^{F-1}\frac{|g_n^{({\rm SI})}[f]|^2}{|g_n^{({\rm D})}[f]|^2}. 
\end{equation}
Thus, the RSI power of the wideband system is derived as
\begin{equation}\label{eq:aveP_wide}
    P_{\rm W} = P^{\rm W}_{n^*_{\rm W}} = \frac{\kappa_y E_{\rm SI}}{F} \sum_{f=0}^{F-1} \frac{|g_{n^*_{\rm W}}^{\rm SI}[f]|^2}{|g_{n^*_{\rm W}}^{\rm D}[f]|^2},
\end{equation}
and the transmission rate of the forward signal is estimated from the capacity as 
\begin{align}
\mathcal{R} & = {\rm BW} \times \mathbb{E}_f \left\{ \log_2 \left(1+ \gamma^{({\rm D})}[f]\right)\right\}\notag\\ 
&= \frac{\rm BW}{F} \times \sum_{f=0}^{F-1} \log_2 \left(1+ \frac{E_s |g_{n^*_{\rm W}}^{\rm D}[f]|^2}{\kappa_y E_{\rm SI}|g_{n^*_{\rm W}}^{\rm SI}[f]|^2}\right),\label{eq:rate_wide}
\end{align}
where ${\rm BW}$ is the system bandwidth.

\section{Performance Analysis}\label{sec:AnalyticalResults}
In this section, we present our analysis of the SIC performance in rich scattering narrowband channels, i.e., the channel model in \eqref{Eq:chan_narr}. We first provide a lower bound on the RSI power of the FAS-assisted SIC in the IBFD system. Subsequently, we employ the first-stage approximation delineated in \cite{Khammassi2023, Xu2024revisiting} to approximate the channel model in  \eqref{Eq:chan_narr}. This allows for a detailed analysis of the RSI power associated with FAS-assisted SIC utilizing the aforementioned approximation.

\subsection{Lower Bound on the Average RSI Power}
Letting $\boldsymbol{y} = [y_0,\dots,y_{N-1}]^T$ and $\boldsymbol{\eta} = [\eta_0,\dots,\eta_{N-1}]^T$, the received signals from all the FAS ports can be collected and written in a vector form as 
\begin{equation}\label{Eq:y_vector}
    \boldsymbol{y} = \boldsymbol{g}^{({\rm D})}s_{\rm D} +\sqrt{\kappa_y}\boldsymbol{g}^{({\rm SI})} \Psi(s_{\rm SI}) +\boldsymbol{\eta}.
\end{equation}
Considering the channel model in \eqref{Eq:chan_narr_vec} and applying the unitary matrix $\boldsymbol{U}$ to $\boldsymbol{y}$, we have
\begin{equation}
    \boldsymbol{y}' = \boldsymbol{U}^\dag \boldsymbol{y} = \boldsymbol{h}^{({\rm D})}s_{\rm D} +\sqrt{\kappa_y}\boldsymbol{h}^{({\rm SI})} \Psi(s_{\rm SI}) +\boldsymbol{\eta}',
\end{equation}
where
\begin{equation}\label{Eq:hh}
    \boldsymbol{h}^{({\rm Z})} = \boldsymbol{U}^\dag \boldsymbol{g}^{({\rm Z})} = \sigma_g \boldsymbol{\Lambda}^{\frac{1}{2}}\boldsymbol{a}^{({\rm Z})}, ~ {\rm Z} \in \{{\rm D},{\rm SI}\},
\end{equation}
and $\boldsymbol{\eta}' = \boldsymbol{U}^\dag \boldsymbol{\eta} \sim \mathcal{CN} (\boldsymbol{0}, \sigma_\eta^2 \boldsymbol{I})$. Following this de-correlated transformation, $\boldsymbol{h}^{({\rm Z})} \sim \mathcal{CN} (\boldsymbol{0},\sigma_g^2 \boldsymbol{\Lambda})$, and the entries in $\boldsymbol{h}^{({\rm Z})}$ are independent of each other. Now, denote 
\begin{equation}\label{Eq:RR_min}
    R' = \min\left\{\frac{|h_0^{({\rm SI})}|^2}{|h_0^{({\rm D})}|^2},\cdots,\frac{|h_{N-1}^{({\rm SI})}|^2}{|h_{N-1}^{({\rm D})}|^2}\right\},
\end{equation}
and denote the average RSI power of the de-correlated channel as $\overline{P}_{lb} = \kappa_y E_{\rm SI} \mathbb{E} \{R'\}$. In the following theorem, we analyze the cumulative distribution function (CDF) of the variable $R'$ and show that $\overline{P}_{lb}$ is the lower bound to the RSI power of the FAS-assisted SIC in the IBFD system. 

\begin{theorem}\label{th:cdf_R}
The CDF of the random variable $R$ in \eqref{Eq:R_min} is upper bounded by CDF of $R'$, given as
\begin{equation}\label{Eq:ubCDFRR}
F_{R'}(r) = 1-\left(\frac{1}{r+1}\right)^N,
\end{equation}
and the average RSI power $\overline{P}$ in \eqref{Eq:aveP_def} is lower bounded by
\begin{equation}\label{Eq:lbaveP}
\overline{P}_{lb} = \frac{\kappa_y E_{\rm SI}}{N-1}.
\end{equation}
\end{theorem}

\begin{proof}
See Appendix \ref{app:ProofTheo1}.
\end{proof}

The random variable $R'$ can be viewed as a special case of $R$ that lacks spatial correlation amongst the FAS ports. In this idealized context, an increase in the number of ports, $N = N_1\times N_2$, results in a reduction of the RSI power, as illustrated in \eqref{Eq:lbaveP}, thereby enhancing SIC. However, in the FAS with $N$ ports uniformly distributed within a limited 2D plane of $W = W_1\lambda \times W_2\lambda$, the FAS ports are usually correlated. Increasing the number of ports decreases the RSI power at the beginning, but then saturates within a fixed $W^* = W^*_1\lambda \times W^*_2\lambda$ size. Accordingly, with a fixed $N$, the average RSI power $\overline{P}$ will approach to the lower bound $\overline{P}_{lb}$ if $W$ increases.

\subsection{Approximated Results}
Using the channel in \eqref{Eq:chan_narr} for analysis results in complicated expressions involving $N$ nested integrals. Thus, the development of a channel model that can effectively approximate the strong correlation inherent in FAS is of great importance. Here, we employ the first-stage approximated channel model as in \cite{Khammassi2023, Xu2024revisiting}, wherein only $M \ll N$ terms in \eqref{Eq:chan_narr} with the largest eigenvalues are considered. The approximated channel, denoted as $\hat{g}_n^{({\rm Z})}$, for ${\rm Z} \in \{{\rm D},{\rm SI}\}$, is defined as
\begin{equation}\label{Eq:g_fsapprox}
\hat{g}_n^{({\rm Z})} = \sigma_g \!\! \sum_{m = 0}^{M-1} \!\! \sqrt{\lambda_m}\mu_{n,m} a_m^{({\rm Z})} \!+\! \sigma_g \sqrt{1\!-\!\sum_{m = 0}^{M-1} \lambda_m \mu^2_{n,m}} b_n^{({\rm Z})},
\end{equation}
where the additional variable $b_n^{({\rm Z})} \sim \mathcal{CN} (0,1)$ is introduced to ensure a constant normalized variance of the channel $\hat{g}^{({\rm Z})}_n$.

Accordingly, the average RSI power in \eqref{Eq:aveP_def} can be approximated as 
\begin{equation}\label{Eq:aveP_fsapprox}
    \overline{P} \approx \kappa_y E_{\rm SI} \mathbb{E}\{\hat{R}\},
\end{equation}
where the random variable $\hat{R}$ is defined as
\begin{equation}
    \hat{R} = \min\left\{\frac{|\hat{g}_0^{({\rm SI})}|^2}{|\hat{g}_0^{({\rm D})}|^2},\cdots,\frac{|\hat{g}_{N-1}^{({\rm SI})}|^2}{|\hat{g}_{N-1}^{({\rm D})}|^2}\right\}.
\end{equation}
In the next theorem, we present the CDF and the expectation of $\hat{R}$ in order to approximate the average RSI power.

\begin{theorem}\label{th:approx}
With the approximation $\hat{g}_n^{({\rm Z})}$ in \eqref{Eq:g_fsapprox}, the CDF and the expectation of $\hat{R}$ are given by \eqref{Eq:CDF_fsapprox} and \eqref{Eq:Exp_fsapprox} {at the top of the next page}, respectively, where $Q_1(\cdot,\cdot)$ denotes the Marcum Q-function of order 1, $I_0(\cdot)$ is the modified Bessel function of the first kind, and $\alpha_n^{({\rm Z})}$ and $\beta_n^{({\rm Z})}$ are given by
\begin{equation}\label{Eq:alphabeta}
    \left\{
        \begin{array}{l}
            \alpha_n^{({\rm Z})} = \sigma_g^2 \left|\sum\limits_{m=0}^{M-1} \sqrt{\lambda_m}\mu_{n,m} a_m^{({\rm Z})}\right|^2, \\
            \beta_n^{({\rm Z})} = \frac{\sigma_g^2}{2} \left(1-\sum\limits_{m=0}^{M-1} \lambda_m \mu_{n,m}^2\right).
        \end{array}
    \right.
\end{equation}
\end{theorem}
\begin{figure*}[tb]
    \begin{equation}\label{Eq:CDF_fsapprox}
        \begin{aligned}
            F_{\hat{R}} = & 1- \frac{1}{\pi^{2M}} \iint_\mathbb{C} \cdots \iint_\mathbb{C}  \exp{\left\{ - \sum_{m=0}^{M-1} \left[ \lvert a_m^{({\rm D})}\rvert^2 + \lvert a_m^{({\rm SI})}\rvert ^2 \right]\right\}} \\ & \!\! \times \!\! \prod_{n=0}^{N-1} \frac{1}{2\beta_n^{({\rm D})}} \!\! \int_0^{+\infty} \!\!\!\! {Q_1 \!\! \left(\sqrt{\frac{\alpha_n^{({\rm SI})}}{\beta_n^{({\rm SI})}}},\sqrt{\frac{rz}{\beta_n^{({\rm SI})}}}\right)} \! \exp{\left(-\frac{\alpha_n^{({\rm D})}+z}{2\beta_n^{({\rm D})}}\right)I_0\left(\frac{\sqrt{\alpha_n^{({\rm D})}z}}{\beta_n^{({\rm D})}}\right)}dz da_0^{({\rm D})} da_0^{({\rm SI})}\cdots da_{M-1}^{({\rm D})} da_{M-1}^{({\rm SI})}
        \end{aligned}
    \end{equation}
    \hrulefill
    \begin{equation}\label{Eq:Exp_fsapprox}
        \begin{aligned}
            \mathbb{E} \{\hat{R}\} = & \frac{1}{\pi^{2M}} \int_0^{+\infty} \iint_\mathbb{C}  \cdots \iint_\mathbb{C}  \exp{\left\{ - \sum_{m=0}^{M-1} \left[ \lvert a_m^{({\rm D})}\rvert^2 + \lvert a_m^{({\rm SI})}\rvert ^2 \right]\right\}} \\ & \!\! \times \!\! \prod_{n=0}^{N-1} \!\! \frac{1}{2\beta_n^{({\rm D})}} \!\! \int_0^{+\infty} \!\!\!\! {Q_1 \!\! \left(\sqrt{\frac{\alpha_n^{({\rm SI})}}{\beta_n^{({\rm SI})}}},\sqrt{\frac{rz}{\beta_n^{({\rm SI})}}}\right)} \! \exp{\left(-\frac{\alpha_n^{({\rm D})}+z}{2\beta_n^{({\rm D})}}\right) \! I_0\left(\frac{\sqrt{\alpha_n^{({\rm D})}z}}{\beta_n^{({\rm D})}}\right)}dz da_0^{({\rm D})} da_0^{({\rm SI})}\!\! \cdots da_{M-1}^{({\rm D})} da_{M-1}^{({\rm SI})} dr
        \end{aligned}
    \end{equation}
    \hrulefill
\end{figure*}

\begin{proof}
See Appendix \ref{app:ProofTheo2}.
\end{proof}

Based on \eqref{Eq:aveP_fsapprox} and Theorem~\ref{th:approx}, we can obtain an approximation of the average RSI power $\overline{P}$ for the FAS-assisted IBFD system, and evaluate the effectiveness of the cancellation mechanism enabled by the FAS-assisted SIC approach.

\section{Simulation Results} \label{sec:Sim}
In this section, we present simulation results to evaluate the performance of FAS-assisted SIC in the IBFD system. For simplicity, we assume that the normalized signal power and normalized channel gains are set to $E_s = 1$ and $\sigma_g^2 = 1$. In addition, we have the normalized constant nonlinear SI power, i.e., $E_{\rm SI} = 1$. The power ratio is set to $\kappa_y = 30~{\rm dB}$. We consider a carrier frequency of $5$ GHz with a wavelength $\lambda = 6$ cm. It is therefore reasonable to consider the normalized FAS size within the interval $[0,5]$, indicating that the FAS size varies from $0~{\rm cm} \times 0~{\rm cm}$ to $30~{\rm cm} \times 30~{\rm cm}$.

Next, the results regarding the rate, the average RSI power, or the CDF can be obtained either through closed-form expressions derived in Section \ref{sec:AnalyticalResults}, or by Monte Carlo simulations from averaging over $10^6$ independent channel realizations.

\subsection{Rich Scattering Narrowband Channels}\label{subsec:SimRSChan}
Here we present the simulation results in rich scattering narrowband channels. We first present the capacity performance of FWS followed by an examination of the accuracy of the proposed bound and the approximation scheme for the RSI power performance. Certain results are obtained from closed-form expressions, including the lower bound derived in \eqref{Eq:lbaveP} and the approximation \eqref{Eq:aveP_fsapprox} of the average RSI power $\overline{P}$. The average capacity and empirical RSI power results on the precise channel model $g^{({\rm Z})}_n$ have been obtained through Monte Carlo simulations over $10^6$ independent channel realizations.

\begin{figure}
\includegraphics[width=\linewidth]{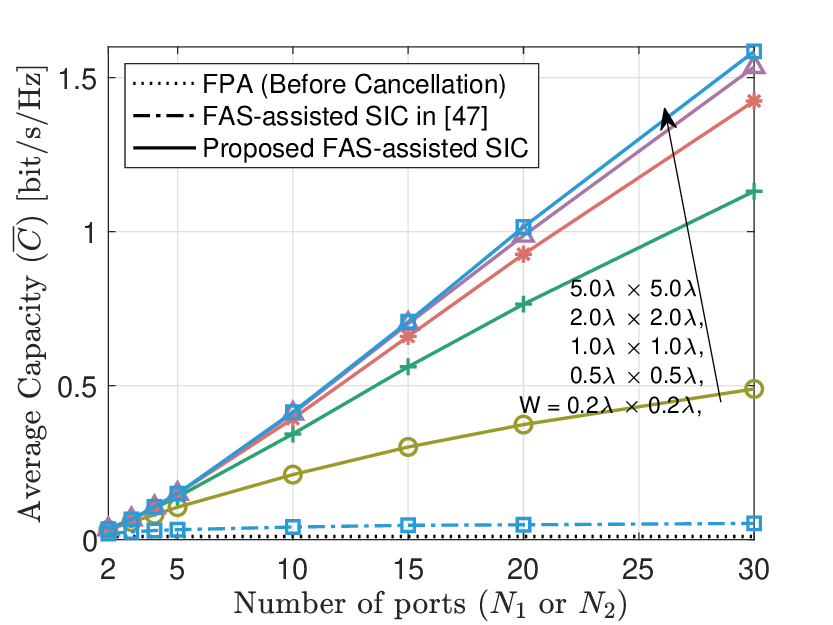}
\caption{Average capacity performance for different FAS configurations with various sizes and resolutions in rich scattering channels.}\label{Fig:RatevsNandW}
\end{figure}

\subsubsection{Capacity Performance}
Fig.~\ref{Fig:RatevsNandW} shows the average capacity of the desired FWS against different FAS configurations, $N$ and $W$. The capacity, computed using \eqref{Eq:achievableRate_Narr}, measures the rate at which the desired signal can be concurrently received during transmission. We compare our proposed FAS-assisted SIC scheme with both FPA-based system and the FAS-assisted method in \cite{sko2023full}. In the FPA system, the fixed antenna is unable to mitigate loopback SI, thereby serving as a benchmark of no SIC. For FAS-assisted method in \cite{sko2023full}, the results are presented with a large FAS plane ($W = 5\lambda \times 5\lambda$), depicted by the blue dash-dotted line with square markers.

The simulation results reveal that the average capacity of the FPA-based system is close to $0~{\rm bit/s/Hz}$, signifying that the desired FWS is entirely overwhelmed by the transmitted LBS before SIC. Though the capacity of the FAS-SIC scheme in \cite{sko2023full} exceeds that of FPA, it remains close to $0~{\rm bit/s/Hz}$. Even when $N = 30\times 30$, the forward signal capacity is approximately $0.05~{\rm bit/z/Hz}$ because \cite{sko2023full} selects the FAS port associated with the maximum forward channel $g^{({\rm D})}$, without considering the loopback SI channel $g^{({\rm SI})}$. Ignoring this SI channel can lead to a selected port with a relatively large SI channel gain, rendering it unsuitable for the IBFD system, particularly when the ratio $\kappa_y$ is large. Conversely, the proposed FAS-assisted SIC approach yields a satisfactory average capacity for the forward signal. Even with a minimal FAS size of $W = 0.2\lambda \times 0.2\lambda$, an average capacity of approximately $0.5~{\rm bit/s/Hz}$ for the forward signal can be realized with a sufficient number of FAS ports (e.g., $N = 30 \times 30$). 

With a fixed FAS size, the average capacity increases alongside the number of FAS ports, $N$, indicating that augmenting the number of ports within a constrained FAS size may contribute positively to enhancing the forward signal capacity in IBFD. Moreover, the results reveal that the capacity also increases with the FAS size, $W$, particularly when $N$ is large. Also, the capacity reaches a saturation point, beyond which the correlation among the FAS ports becomes negligible and the capacity is predominantly limited by the finite number of FAS ports. For instance, in Fig.~\ref{Fig:RatevsNandW}, the blue solid curve representing $W = 5\lambda \times 5\lambda$ exhibits marginal performance improvement compared to the purple curve with $W = 2\lambda\times 2\lambda$. This suggests that a typical size of FAS, feasible for a user terminal, could facilitate IBFD in rich scattering channels.

\begin{figure}
\centering
\includegraphics[width = \linewidth]{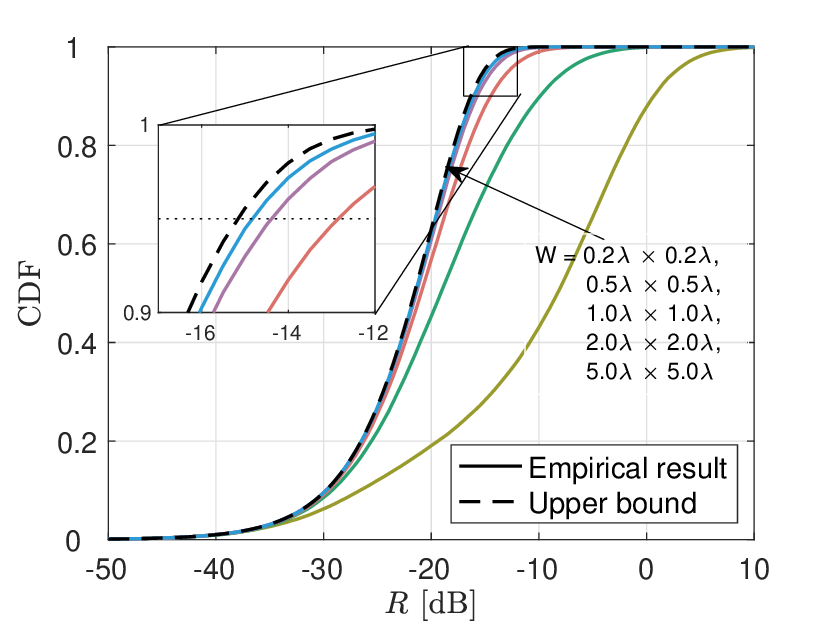}
\caption{The CDF of $R_{\min}$ with $N = 10 \times 10$ FAS ports.}\label{Fig:CDFofRmin}
\end{figure}

\subsubsection{RSI Power and its Lower Bound} 
In Fig.~\ref{Fig:CDFofRmin}, the CDF of the random variable $R$ along with its upper bound in \eqref{Eq:ubCDFRR} is provided. The results demonstrate that as the size of FAS, $W$, increases, the empirical curves converge towards the upper bound. This phenomenon occurs because when $N$ is fixed and $W$ is expanded, the distance between the adjacent ports increases, reducing the correlation among the ports.

The empirical results and the lower bound of RSI power are compared in Fig.~\ref{Fig:PRSIvsNandW}. The average RSI power, $\overline{P}$, is defined in \eqref{Eq:aveP_def}, with its lower bound calculated according to \eqref{Eq:lbaveP}. The results of the FPA-based system and the FAS-assisted SIC in \cite{sko2023full} are presented as benchmarks. Here, the results of the FPA-based system serves as the reference RSI power before cancellation. The disparity between this power and the results obtained from the FAS-assisted SIC can be interpreted as the cancellation capability.

As can be seen, the FAS-assisted SIC in \cite{sko2023full} with $N = 30 \times 30$ FAS ports can provide approximately $12$ to $20~{\rm dB}$ of SIC. Notably, its cancellation capability with a large size of $W = 5\lambda \times 5\lambda$ is lower than that of the proposed scheme with a much smaller size of $W = 0.5\lambda \times 0.5\lambda$. Comparing to the FPA-based scheme, the proposed FAS-assisted SIC demonstrates considerable interference cancellation. Even with a small size of $W = 0.2\lambda \times 0.2\lambda$, an approximately SIC of $15~{\rm dB}$ can be achieved. Moreover, the cancellation capability improves with the FAS size, $W$. The empirical results tend to approach the lower bound as the FAS size becomes large as well. 

For a fixed FAS size $W$, the RSI power exhibits a decreasing trend with an increasing number of ports until a saturation point is reached, beyond which additional ports yield no further enhancement in SIC. For example, when the physical size is established as $W = 5\lambda \times 5\lambda$, as represented by the blue solid curve, a SIC of $17~{\rm dB}$ is observed when $N = 2 \times 2$, increasing to approximately $40~{\rm dB}$ for $N = 30 \times 30$. The number of ports necessary to reach saturation becomes higher for larger sizes. The yellowish green solid curve indicates that the saturation is reached at $N^* = 3 \times 3$ for $W =0.2\lambda \times 0.2\lambda$, while it extends to approximately $N^* = 15\times 15$ for $W =0.5\lambda \times 0.5\lambda$ as shown by the green solid curve. 

With a fixed number of ports $N$, the RSI power decreases as the FAS size $W$ increases, particularly when the number of FAS ports $N$ is deemed sufficient. The SIC capability will also reach a saturation point, denoted as $W^*$, for a fixed $N$, with the average RSI power at this saturation point approximated by the lower bound. Moreover, this saturation point is relatively easy to attain with a size close to $W^* = 2\lambda \times 2 \lambda$. In essence, the RSI power for the FAS size ranging from $W = 2\lambda \times 2 \lambda$ to $5\lambda \times 5\lambda$ is comparable and closely approximates the lower bound, suggesting that a regular physical size of FAS is sufficient to provide substantial SIC in rich scattering channels. An understanding of this phenomenon can be gleaned from the lower bound in \eqref{Eq:lbaveP}, as the lower bound power is inversely related to $N$, thereby rendering the power in dBm inversely logarithmic related to $N$. Consequently, achieving a significant reduction in the lower bound of average RSI power necessitates a considerable increase in $N$.

\begin{figure}
\begin{center}
\includegraphics[width = \linewidth]{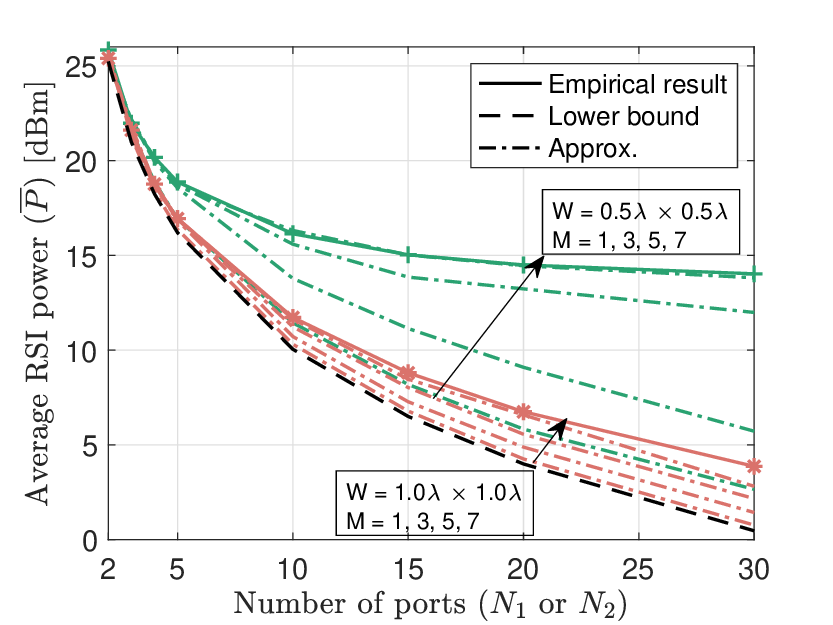}
\caption{Average RSI power and its approximation for different FAS configurations with various sizes and resolutions in rich scattering channels.}\label{Fig:FSApprox_PRSIvsN}
\end{center}
\end{figure}

\subsubsection{Approximated Results of RSI power}
In Fig.~\ref{Fig:FSApprox_PRSIvsN}, we study the performance of the approximation calculated as in \eqref{Eq:aveP_fsapprox}. Although the channel model in \eqref{Eq:chan_narr} has been significantly simplified through the introduction of $\hat{g}_n^{({\rm Z})}$ in \eqref{Eq:g_fsapprox}, the CDF and the expectation detailed in Theorem \ref{th:approx} remain challenging to compute due to the presence of $2M$-fold surface integrals. Therefore, the CDF and expectation in Theorem \ref{th:approx} are calculated based on numerical integral of $\hat{\boldsymbol{a}}^{({\rm D})}$ and $\hat{\boldsymbol{a}}^{({\rm SI})}$.

It is observed that in comparison to the empirical results, the approximated results serve as lower bounds that are tighter than the lower bound in \eqref{Eq:lbaveP}. As the value of $M$ increases, the bounds converge toward the empirical results. For FAS with $W = 0.5\lambda \times 0.5\lambda$, the approximation is deemed sufficiently accurate when $ M = 7$. However, a small gap persists between the approximation and the empirical results for the RSI power when $M=7$ for the case of $W = 1.0\lambda \times 1.0\lambda$ and $N = 30 \times 30$. In this case, a higher approximation level of $M$ is necessary to achieve the desired accuracy, since the expanding area of the FAS leads to a more complex rank correlation matrix $\boldsymbol{\Sigma}$, with additional non-negligible eigenvalues.

\begin{figure}[t]
\begin{center}
\includegraphics[width = \linewidth]{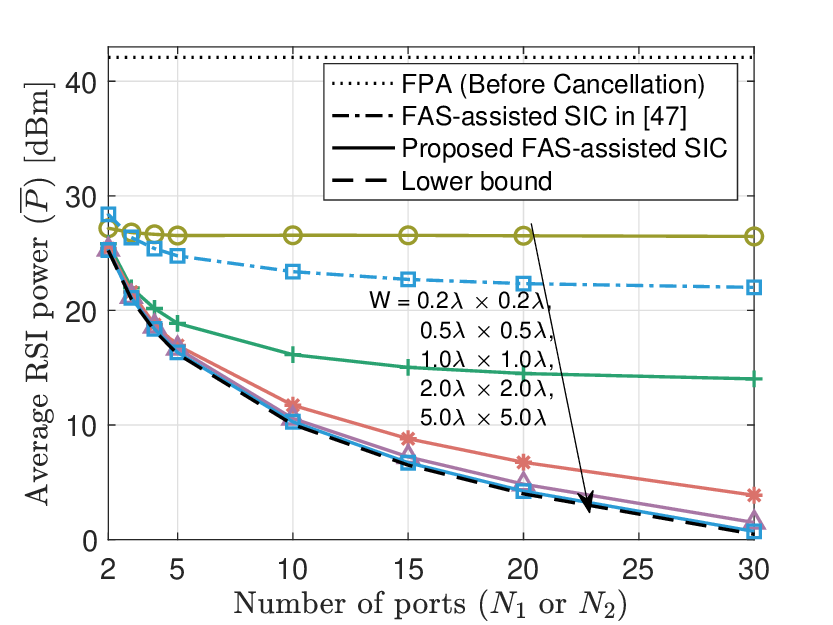}
\caption{Average RSI power for different FAS configurations with various sizes and resolutions in rich scattering channels.}\label{Fig:PRSIvsNandW}
\end{center}
\end{figure}

\subsection{Impact of Channel Estimation}
The aforementioned simulations are conducted under the assumption of perfect CSI. In this subsection, the impact of the channel estimation on the performance is evaluated. A linear minimum-mean-squared-error (LMMSE)-based channel estimation process, as described in \cite{sko2023full}, is utilized. Specifically, $L_e$ pilot-training symbols per block are employed to estimate the forward channels $\boldsymbol{g}^{({\rm D})}$, and additional $L_e$ pilot-training symbols are used to estimate the loopback channels $\boldsymbol{g}^{({\rm SI})}$. Note that these $L_e$ training symbols for loopback channel estimation introduce additional overhead compared to the approach in~\cite{sko2023full}. Considering the training overhead and taking this into account, the capacity in \eqref{Eq:achievableRate_Narr} is recalculated as
\begin{equation}
C_{\rm CE} = \left(1-\frac{2L_e}{L_c}\right)\log_2\left(1+\frac{E_s}{\kappa_yE_{\rm SI}R}\right),
\end{equation}
where $L_c$ is the length of channel coherence block. Consistent with \cite{sko2023full}, $L_c = 5\times 10^6$ is adopted for the simulations.

\begin{figure}
\includegraphics[width = \linewidth]{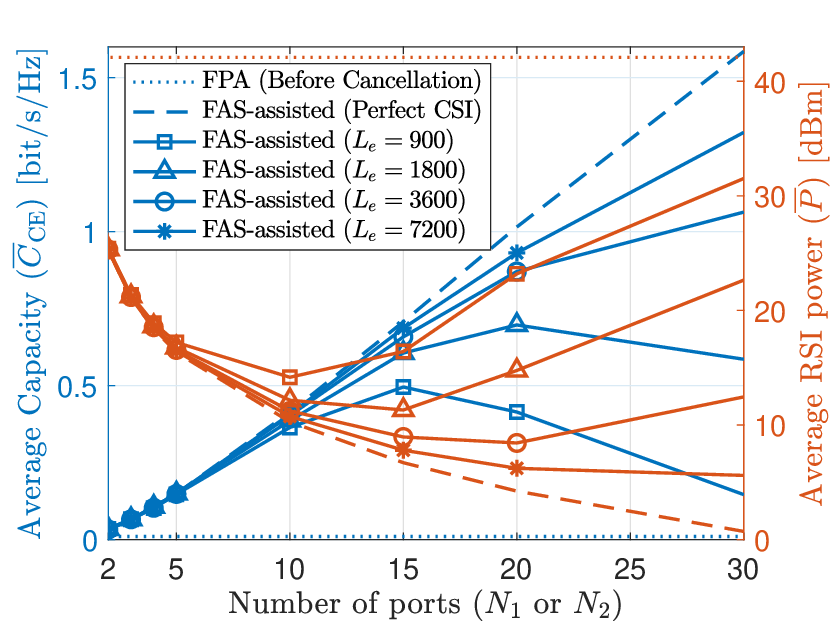}
\caption{Average RSI power and average rate for different FAS configurations with various resolutions for different length of channel estimation symbols, $L_e$, assuming the FAS size of $W = 5\lambda \times 5\lambda$.}\label{Fig:RatevsN_CE}
\end{figure}

Fig.~\ref{Fig:RatevsN_CE} illustrates the results of average RSI power and rate performance for varying the numbers of pilot-training symbols $L_e = \{900, 1800, 3600, 7200\}$. When utilizing a limited number of pilot-training symbols, such as $L_e = \{900, 1800\}$, the addition of extra ports results in a clear reduction in FAS-assisted SIC performance. This occurs because the allocation of pilot-training symbols for channel estimation of each port diminishes, thereby impairing the quality of channel estimation and jeopardizing the SIC performance. Conversely, by increasing the number of pilot-training symbols, for instance to $L_e = \{3600, 7200\}$, an adequate number of symbols can be dedicated to channel estimation of all the FAS ports. As a result, the forward rate experiences an increase, while the RSI power decreases, approaching to the ideal scenario with perfect CSI. This observation indicates that with a sufficient number of training symbols, an enhanced receive diversity gain can be achieved through an increased number of FAS ports, thereby improving the SIC performance. Nevertheless, in the scenarios with a limited quantity of pilot-training symbols, further increasing the number of FAS ports may result in performance degradation in a practical setting.

\begin{figure}
\begin{center}
\includegraphics[width = \linewidth]{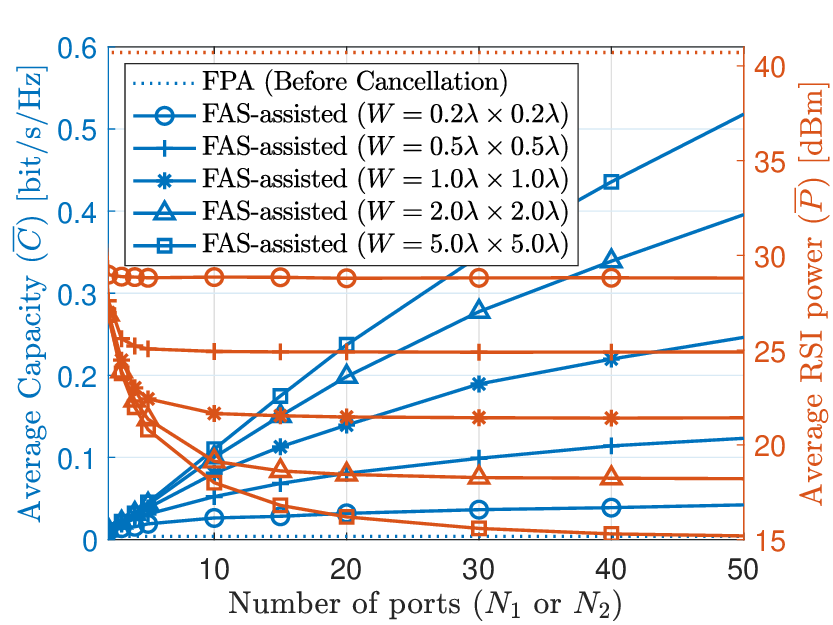}
\caption{Average RSI power and average capacity comparison for different FAS configurations with various sizes and resolutions in finite-scattering channel with $(K,N_p) = (3,2)$.} \label{Fig:RandPRSIvsNnW_FSChan} 
\end{center}
\end{figure}

\subsection{Finite-Scattering Narrowband Channels}
In many practical scenarios, a LOS exists between the co-located Tx and Rx at the BS. For this reason, here we provide Monte-Carlo simulations utilizing the finit-scatterer channel model for the loopback channel, as defined in \eqref{Eq:FSChann}. The results are presented in Fig.~\ref{Fig:RandPRSIvsNnW_FSChan}, where the Rice factor and the number of scattered components are set as $K = 3$ and $N_p = 2$.

A comparison of these results with those in rich scattering channels, as presented in Figs.~\ref{Fig:RatevsNandW} and \ref{Fig:PRSIvsNandW}, reveals a consistent trend that the average rate exhibits an increase according to the number of FAS ports, $N$, and the normalized FAS size, $W$. Additionally, the RSI power tends to decrease as $N$ or $W$ increases. This implies that the FAS-assisted SIC approach is again capable of providing interference cancellation in finite-scattering channels, particularly with larger FAS configurations. Nevertheless, the performance is notably diminished in the finite-scattering channel case. In particular, the FAS-assisted SIC achieves approximately $12~{\rm dB}$ of SIC, while the forward rate falls below $0.05~{\rm bit/s/Hz}$ with $W = 0.2\lambda \times 0.2\lambda$. Upon increasing the FAS configuration (either $N$ or $W$), the cancellation achieved by the FAS-assisted approach with a size of $W = 2\lambda \times 2\lambda$ reaches approximately $22.5~{\rm dB}$, even with a large number of ports ($N = 50\times 50$), and reaches approximately $25.5~{\rm dB}$ with $W = 5\lambda \times 5\lambda$. The forward rate remains relatively constrained in the finite-scatterer channel, measuring approximately $0.35~{\rm bit/s/Hz}$ with a FAS configuration of $N = 30\times 30$ and $W = 5\lambda \times 5\lambda$, and $0.52~{\rm bit/s/Hz}$ as the number of FAS ports increases to $N = 50\times 50$.

\begin{figure}
\begin{center}
\includegraphics[width = \linewidth]{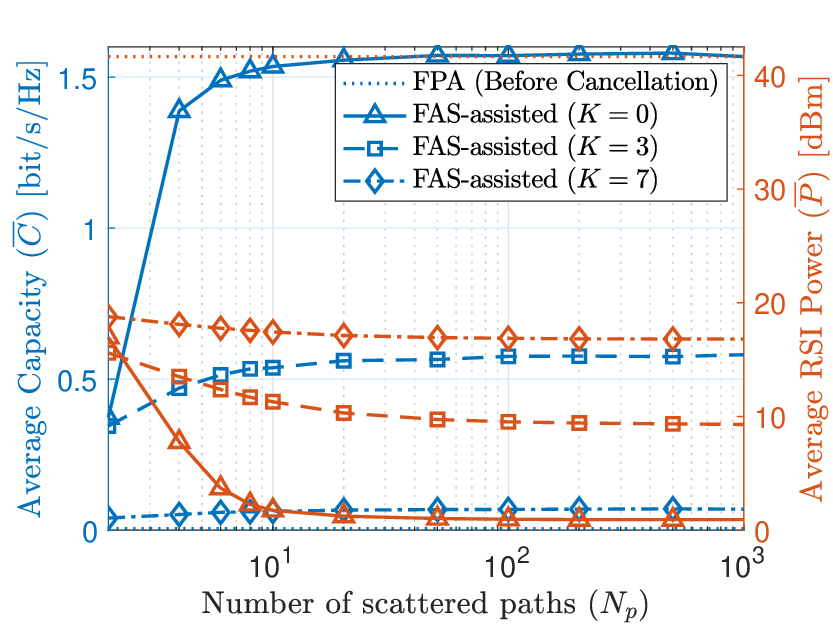}
\caption{Average RSI power and average capacity comparison for different channel condition $(K,N_p)$, in finite-scattering channels, with FAS configuration of $N = 30 \times 30$ and $W = 5\lambda \times 5\lambda$.}\label{Fig:RandPRSIvsNp_FSChan}
\end{center}
\end{figure}

Now we access the impact of the Rice factor, $K$, and the number of scattered components, $N_p$, on the SIC performance. Fig.~\ref{Fig:RandPRSIvsNp_FSChan} presents the results for capacity and RSI power in the finite-scatterer channels under various channel conditions, with a FAS configuration of $N = 30 \times 30$ and $W = 5\lambda \times 5\lambda$. The results align closely with those in Figs.~\ref{Fig:RatevsNandW} and \ref{Fig:PRSIvsNandW} when LOS is absent and multipath is rich, i.e., $K=0$ and a large $N_p$, indicating that the SIC capability in rich scattering channels may serve as the upper bound for the FAS-assisted SIC in the IBFD system. Nevertheless, the RSI power increases with a decrease in $N_p$ and an increase in $K$, while the forward rate diminishes under the same conditions. This illustrates that the SIC capability of the FAS-assisted approach deteriorates in strong LOS environments. Nonetheless, FAS-assisted SIC can still achieve approximately $20~{\rm dB}$ of cancellation, even when $K=7$. Also, we observe that the performance is more significantly influenced by $N_p$ when the Rice factor $K$ is small, as the LOS component dominates when $K$ is large.

\subsection{Wideband Channels}
In this subsection, we attempt to investigate the performance of FAS-assisted SIC in wideband channels. An orthogonal frequency division multiplexing (OFDM) system which operates with a bandwidth of $5~{\rm MHz}$ coupled with a $512$-point FFT, i.e., $F = 512$, is considered. The sampling frequency is calculated as $512\times 15\times 10^3 = 7.68 \times 10^6$ samples/sec, with a sampling duration of approximately $0.13~{\rm \mu s}$. Considering a LOS environment between the co-located Tx and Rx, the loopback channel is modeled as Rician fading, which is applicable in wireless backhaul scenarios. Specifically, we utilize a 3-tap model, with parameters derived from channel calculations by ray-tracing for a specific location pertinent to IAB \cite{suk2022full,Hong2023frequency} and represent the correlation among the FAS ports using the received correlation matrix $\boldsymbol{\Sigma}$ with $N\times N$ elements. The power and delay profile of this model are characterized as 
$$\text{[path power (dB)/path delay (\(\rm \mu s\))]:}~[0/0,-25/2,-30/7].$$
This channel model features a maximum delay spread of $7~{\rm \mu s}$, equating to $54$ OFDM samples in our simulation parameters.

\begin{figure}
\begin{center}
\includegraphics[width = \linewidth]{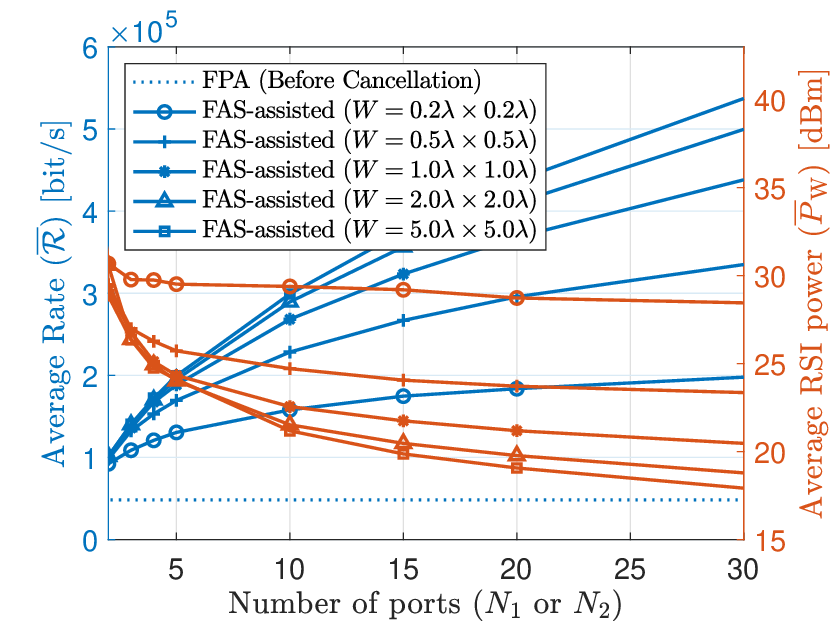}
\caption{Average RSI power comparison for different FAS configurations with various sizes and resolutions in the wideband IAB channel with $K=3$.}\label{Fig:RnPRSIvsNnW_Wide}
\end{center}
\end{figure}

Fig.~\ref{Fig:RnPRSIvsNnW_Wide} presents the simulation results for the wideband IAB channel. In these scenarios, the FAS-assisted SIC selects the FAS port using \eqref{Eq:PortSelection_wide}. Consequently, the average RSI power and the FWS rate are calculated as \eqref{eq:aveP_wide} and \eqref{eq:rate_wide}, respectively. 

The results indicate similarities to those observed in finite-scattering channels, as presented in Fig.~\ref{Fig:RandPRSIvsNnW_FSChan}. The FAS-assisted SIC demonstrates an interference cancellation range from $14$ to $25~{\rm dB}$ within the wideband channel, with configurations of $N = 30 \times 30$ FAS ports distributed over $W$ ranging from $0.2\lambda \times 0.2\lambda$ to $5\lambda \times 5\lambda$. The RSI power exhibits a decline with the increase in FAS configurations, $N$ or $W$, suggesting that greater configurations enhance cancellation capabilities of FAS-assisted SIC in IBFD. Notably, the transmission rate of the forward signal significantly benefits from the FAS-assisted SIC compared to the FPA-based system, with an improved rate. However, performance gains associated with increasing FAS size become marginal when transitioning from $W = 2\lambda \times 2 \lambda$ to $5\lambda \times 5\lambda$. As discussed in Fig. \ref{Fig:PRSIvsNandW}, this is because the cancellation level is approaching its upper bound. 

\begin{figure}
\begin{center}
\subfigure[]{\includegraphics[width = \linewidth]{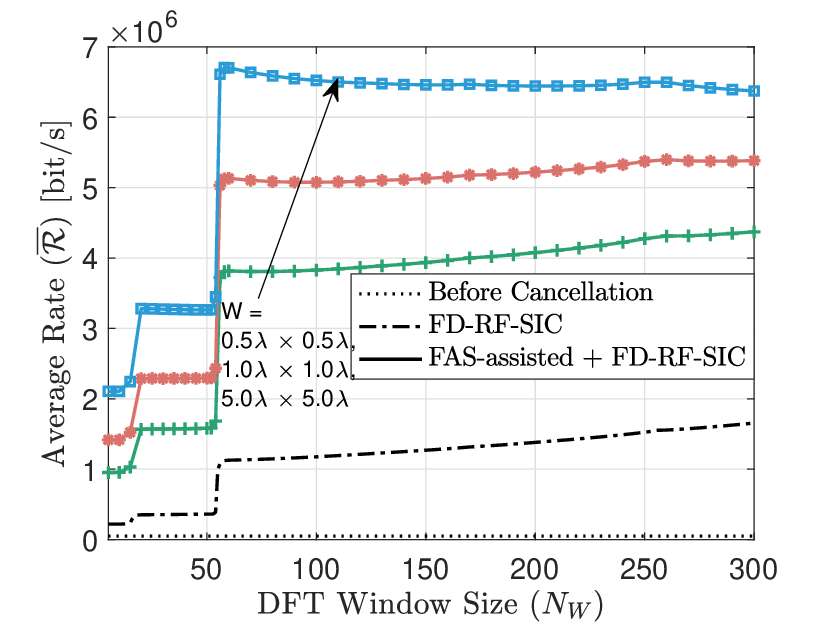}}
\subfigure[]{\includegraphics[width = \linewidth]{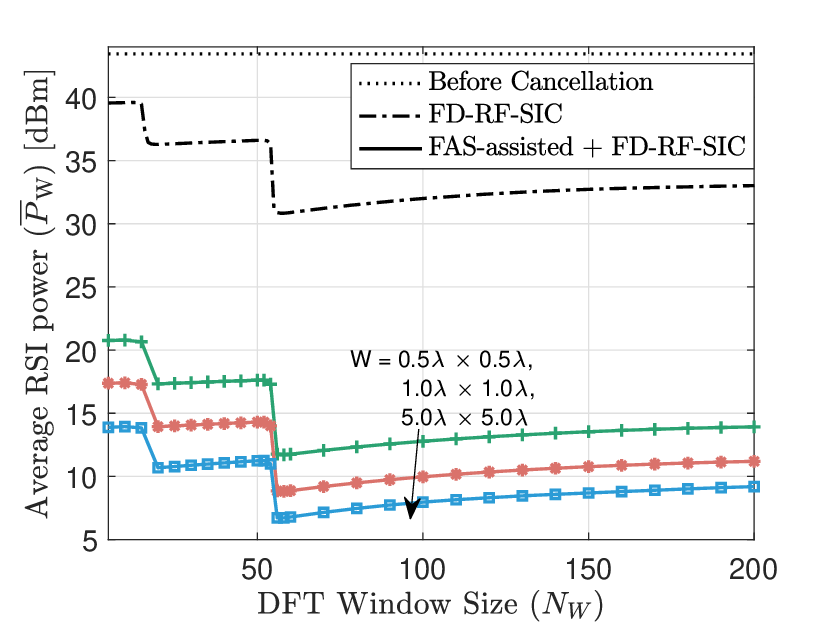}}
\caption{The performance of FAS-assisted SIC with $N = 30 \times 30$, combined with FD-RF-SIC in \cite{Hong2023frequency}. Results are presented for (a) average rate, and (b) average RSI power against the DFT window size, $N_W$.}\label{Fig:2SIC}
\end{center}
\end{figure}

\subsection{Combining with Frequency-Domain SIC}\label{subsec:combination}
As FAS-assisted SIC is implemented during FAS reception, it is typically compatible with other active SIC techniques. In this subsection, we explore the integration of the proposed FAS-assisted SIC with the recent proposed FD-RF-SIC in \cite{Hong2023frequency}. The FD-RF-SIC employs an RF reference signal (RERS) to calculate the SIC filter weights, utilizing DFT-windowing to minimize the total mean square error (MSE). The RFRS can be either an over-the-air RFRS (OTA-RFRS) or a direct RFRS (D-RFRS). The results obtained from combining FAS-assisted SIC and FD-RF-SIC are illustrated in Fig.~\ref{Fig:2SIC}, where the simulations utilize the D-RFRS for FD-RF-SIC. The FAS configuration is established with $N = 30 \times 30$ ports among various size of $W = 0.5\lambda \times 0.5\lambda$, $1\lambda \times 1\lambda$, or $5\lambda \times 5 \lambda$. 

As seen, the performance of FD-RF-SIC in \cite{Hong2023frequency}  improves with the increase of DFT size, $N_W$. Optimal performance is achieved when the window size exceeds the maximum delay spread. Additionally, the performance of FD-RF-SIC experiences two notable enhancements, which correspond to the delay profile of the IAB channels. With a sufficient size of DFT window, the FD-RF-SIC achieves approximately $12.5~{\rm dB}$ of cancellation, with a forward signal rate of approximately $1~{\rm Mbps}$. The integration of both SIC schemes can further diminish the RSI power, yielding an additional cancellation level. The combination of the two schemes produces an approximately $36.7~{\rm dB}$ of cancellation, with a forward signal rate of $7~{\rm Mbps}$, achieved under the FAS configuration of $N = 30\times 30$ and $W = 5\lambda \times 5\lambda$. This significant improvement underscores the efficacy of SIC capability and the compatibility of the proposed FAS-assisted SIC scheme. 

\section{Conclusion}\label{sec:conclusion}
In this paper, we introduced a FAS-assisted SIC framework for IBFD systems. This framework utilizes the FAS technology at the receiver to naturally identify the interference-free port. We derived analytical expressions for the lower bound and an approximation of the RSI power after FAS-assisted SIC under rich-scattering channels. Simulation results validated the approximation. We also considered more realistic conditions, accounting for channel estimation and situations under practical finite-scattering channels and IAB channels, as well as considering the integration with other RF-SIC schemes. The results demonstrated remarkable cancellation capabilities of the proposed FAS-assisted SIC scheme, achieving approximately $40~{\rm dB}$ of SIC and a forward signal rate surpassing $1.5~{\rm bit/s/Hz}$ in rich-scattering channels. In a more realistic IAB channel model, approximately $25~{\rm dB}$ of SIC along with a forward rate exceeding $5\times 10^5 ~{\rm bit/s}$ can be achieved with a $5~{\rm MHz}$ bandwidth. Additionally, the proposed FAS-assisted SIC scheme demonstrated excellent compatibility with other RF-SIC approaches, and a combined methodology of FAS-assisted SIC and FD-RF-SIC yielded approximately $37~{\rm dB}$ of SIC and a forward rate of $6.8\times 10^6~{\rm bit/s}$. 

\appendices
\section{Proof of Theorem \ref{th:cdf_R}}\label{app:ProofTheo1}
From \eqref{Eq:hh}, we can derive that
\begin{equation}
    \frac{|h_n^{({\rm SI})}|^2}{|h_n^{({\rm D})}|^2} = \frac{\sigma_g \lambda_n|a_n^{({\rm SI})}|^2}{\sigma_g \lambda_n|a_n^{({\rm D})}|^2} = \frac{|a_n^{({\rm SI})}|^2}{|a_n^{({\rm D})}|^2}.
\end{equation}
For all $0\leq n <N$ and ${\rm Z} \in \{{\rm D},{\rm SI}\}$, it can be seen that $a_n^{({\rm Z})}$ are independent of each other, $|a_n^{({\rm Z})}| \sim {\rm Rayleigh}\left(1/\sqrt{2}\right)$, and $|a_n^{({\rm Z})}|^2 \sim {\rm Exponential}\left(1\right)$. The random variable in \eqref{Eq:RR_min} can be reformulated as 
\begin{equation}\label{Eq:RRR_min}
    R' = \min\left\{\frac{|a_0^{({\rm SI})}|^2}{|a_0^{({\rm D})}|^2},\cdots,\frac{|a_{N-1}^{({\rm SI})}|^2}{|a_{N-1}^{({\rm D})}|^2}\right\}. 
\end{equation}
Evidently, $\kappa_y E_{\rm SI} R'$, with $R'$ in \eqref{Eq:RRR_min}, can be interpreted as the average RSI power in an uncorrelated channel represented by
\begin{equation}\label{Eq:y_vector_uc}
    \tilde{\boldsymbol{y}} = \boldsymbol{a}^{({\rm D})}s_{\rm D} +\sqrt{\kappa_y} \boldsymbol{a}^{({\rm SI})}\Psi (s_{\rm SI}) + \tilde{\boldsymbol{\eta}}.
\end{equation}
This uncorrelated channel can additionally be regarded as a special case of that in \eqref{Eq:y_vector}, applicable when ports are sufficiently distanced such that the channel gains across different ports are uncorrelated. Therefore, the CDF of $R'$ serves as an upper bound for that of $R$, such that $\overline{P}_{lb} = \kappa_y E_{\rm SI} \mathbb{E}\{R'\}$ is the lower bound for $\overline{P} = \kappa_y E_{\rm SI} \mathbb{E}\{R\}$.

The CDF of $R'$ can be calculated as
\begin{align}
F_{R'}(r) & = {\rm Pr}\left\{\min\left\{\frac{|a_0^{({\rm SI})}|^2}{|a_0^{({\rm D})}|^2},\cdots,\frac{|a_{N-1}^{({\rm SI})}|^2}{|a_{N-1}^{({\rm D})}|^2}\right\} \leq r\right\}\notag\\ 
& = 1- {\rm Pr}\left\{\frac{|a_0^{({\rm SI})}|^2}{|a_0^{({\rm D})}|^2} > r, \dots, \frac{|a_{N-1}^{({\rm SI})}|^2}{|a_{N-1}^{({\rm D})}|^2} > r \right\}\notag\\ 
& = 1 - \left(1-{\rm Pr}\left\{\frac{\hat{X}}{X} \leq r \right\} \right)^N,\label{Eq:CDFRR_min_proof}
\end{align}
where $\hat{X}$, $X\sim {\rm Exponential}(1)$ and they are independent. Hence, the CDF of $\hat{X}/X$ can be derived as 
\begin{align}
{\rm Pr}\left\{\frac{\hat{X}}{X} \leq r \right\} & \hspace{.5mm}= {\rm Pr} \left\{\hat{X}<rX \right\}\notag\\ 
& \hspace{.5mm}= \int_0^{+\infty} {F_{\hat{X}}\left(rx\right) f_X\left(x\right)}dx \notag\\ 
& \overset{(a)}{=} \int_0^{+\infty} (1-\exp(-rx))\exp(-x)dx\notag\\ 
& \hspace{.5mm}= 1-\frac{1}{r+1},\label{Eq:ExpDivExp}
\end{align}
where $(a)$ uses the CDF and PDF of exponential distribution for $\hat{X}$ and $X$. By substituting \eqref{Eq:ExpDivExp} into \eqref{Eq:CDFRR_min_proof}, we have \eqref{Eq:ubCDFRR}.

The expectation of $R'$ can be calculated as
\begin{align}
\mathbb{E}\left\{ R' \right\} & = \int_0^{+\infty} {\left[1-F_{R'}(r)\right]} dr\notag\\ 
& = \int_0^{+\infty} {\left(\frac{1}{r+1}\right)^N} dr\notag\\ 
& = \frac{1}{N-1}.\label{Eq:expRR_min}
\end{align}
By substituting \eqref{Eq:expRR_min} into $\overline{P}_{lb} = \kappa_y E_{\rm SI} \mathbb{E}\{R'\}$, we obtain \eqref{Eq:lbaveP}.

\section{Proof of Theorem \ref{th:approx}}\label{app:ProofTheo2}
It is established from \eqref{Eq:g_fsapprox} that the approximated channel $\hat{g}^{({\rm Z})}_n$ adheres to a Gaussian distribution. Given a specific $\hat{\boldsymbol{a}}^{({\rm Z})} = [a_0^{({\rm Z})},\dots,a_{M-1}^{({\rm Z})}]^T$, we have
\begin{equation}
\hat{g}^{({\rm Z})}_n \sim \mathcal{CN} \left(\sigma_g \sum_{m=0}^{M-1} \sqrt{\lambda_m} \mu_{n,m} a_m^{({\rm Z})}, 2\beta_n^{({\rm Z})} \right),
\end{equation}
where $\beta_n^{({\rm Z})}$ is given in \eqref{Eq:alphabeta}. Consequently, $\lvert \hat{g}^{({\rm Z})}_n \rvert$ follows Rice distribution as
\begin{equation}
\lvert \hat{g}^{({\rm Z})}_n\rvert \sim {\rm Rice} \left( \sqrt{\alpha_n^{({\rm Z})}}, \sqrt{\beta_n^{({\rm Z})}} \right),
\end{equation}
where $\alpha_n^{({\rm Z})}$ and $\beta_n^{({\rm Z})}$ are given in \eqref{Eq:alphabeta}. We now consider the random variable of $\lvert \hat{g}^{({\rm Z})}_n \rvert ^2$ conditioned on $\hat{\boldsymbol{a}}^{({\rm Z})}$. This variable follows a non-central Chi-square distribution, with CDF and PDF expressed, respectively, as 
\begin{equation}
F_{\lvert \hat{g}^{({\rm Z})}_n\rvert^2 | \hat{\boldsymbol{a}}^{({\rm Z})}}(r) = 1-Q_1\left(\sqrt{\frac{\alpha_n^{({\rm Z})}}{\beta_n^{({\rm Z})}}},\sqrt{\frac{r}{\beta_n^{({\rm Z})}}}\right),
\end{equation}
and
\begin{equation}
f_{\lvert \hat{g}^{({\rm Z})}_n\rvert^2 | \hat{\boldsymbol{a}}^{({\rm Z})}}\!(r) = \frac{1}{2\beta_n^{({\rm Z})}} \!\! \exp\!\!{\left(\!\!-\frac{\alpha_n^{({\rm Z})}+r}{2\beta_n^{({\rm Z})}}\!\right)\!\! I_0\!\!\left(\!\!\frac{\sqrt{\alpha_n^{({\rm Z})}r}}{\beta_n^{({\rm Z})}}\right)}.
\end{equation}
Note that $\lvert \hat{g}^{({\rm D})}_n\rvert^2$ and $\lvert \hat{g}^{({\rm SI})}_n\rvert^2$ are independent. Therefore, the random variable $\hat{R}_n = \lvert \hat{g}^{({\rm SI})}_n\rvert^2 / \lvert \hat{g}^{({\rm D})}_n\rvert^2$ conditioned on $(\hat{\boldsymbol{a}}^{({\rm D})},\hat{\boldsymbol{a}}^{({\rm SI})})$ represents the ratio of these two independent random variables. Its CDF can then be derived as
\begin{align}
F_{\hat{R}_n}&_{|(\hat{\boldsymbol{a}}^{({\rm D})},\hat{\boldsymbol{a}}^{({\rm SI})})}(r) \notag\\ 
= & \int_0^\infty F_{\lvert \hat{g}^{({\rm SI})}_n\rvert^2|\hat{\boldsymbol{a}}^{({\rm SI})}}(rz) f_{\lvert \hat{g}^{({\rm D})}_n\rvert^2|\hat{\boldsymbol{a}}^{({\rm D})}}(z) dz \notag\\ 
= & 1-\frac{1}{2\beta_n^{({\rm D})}} \int_0^{+\infty} {Q_1\left(\sqrt{\frac{\alpha_n^{({\rm SI})}}{\beta_n^{({\rm SI})}}},\sqrt{\frac{rz}{\beta_n^{({\rm SI})}}}\right)}\notag\\ 
& \quad\times \exp{\left(-\frac{\alpha_n^{({\rm D})}+z}{2\beta_n^{({\rm D})}}\right)I_0\left(\frac{\sqrt{\alpha_n^{({\rm D})}z}}{\beta_n^{({\rm D})}}\right)}dz.
\end{align}
For a specified pair $(\hat{\boldsymbol{a}}^{({\rm D})},\hat{\boldsymbol{a}}^{({\rm SI})})$, the independent variables $b_n^{({\rm Z})}, \forall n \in \{0,\dots, N-1\}$ in \eqref{Eq:g_fsapprox} indicate that $\hat{R}_n$ are independent of each other. Consequently, the CDF of $\hat{R} = \min\{\hat{R}_0,\dots, \hat{R}_{N-1}\}$ can be calculated as
\begin{align}
F_{\hat{R}}& _{|(\hat{\boldsymbol{a}}^{({\rm D})},\hat{\boldsymbol{a}}^{({\rm SI})})}(r) \notag\\ 
= & 1- \prod_{n=0}^{N-1}\left[1-F_{\hat{R}_n|(\hat{\boldsymbol{a}}^{({\rm D})},\hat{\boldsymbol{a}}^{({\rm SI})})}(r)\right] \notag\\ 
= & 1- \prod_{n=0}^{N-1} \frac{1}{2\beta_n^{({\rm D})}} \int_0^{+\infty} {Q_1\left(\sqrt{\frac{\alpha_n^{({\rm SI})}}{\beta_n^{({\rm SI})}}},\sqrt{\frac{rz}{\beta_n^{({\rm SI})}}}\right)} \notag\\ 
& \quad\times \exp{\left(-\frac{\alpha_n^{({\rm D})}+z}{2\beta_n^{({\rm D})}}\right)I_0\left(\frac{\sqrt{\alpha_n^{({\rm D})}z}}{\beta_n^{({\rm D})}}\right)}dz.\label{Eq:conPDF_hatR}
\end{align}
As a result, the CDF of $\hat{R}$ can be considered as the expectation of \eqref{Eq:conPDF_hatR} over $(\hat{\boldsymbol{a}}^{({\rm D})},\hat{\boldsymbol{a}}^{({\rm SI})})$, resulting in \eqref{Eq:CDF_fsapprox}.
Similarly, based on the CDF in \eqref{Eq:conPDF_hatR}, the expectation of $\hat{R}$ conditioned on $(\hat{\boldsymbol{a}}^{({\rm D})},\hat{\boldsymbol{a}}^{({\rm SI})})$ is expressed as
\begin{align}
\mathbb{E} \{& \hat{R}| (\hat{\boldsymbol{a}}^{({\rm D})},\hat{\boldsymbol{a}}^{({\rm SI})}) \} \notag\\ 
= & \int_{0}^{+\infty} 1- F_{\hat{R}| (\hat{\boldsymbol{a}}^{({\rm D})},\hat{\boldsymbol{a}}^{({\rm SI})})}(r) dr \notag\\ 
= & \int_{0}^{+\infty} \prod_{n=0}^{N-1} \frac{1}{2\beta_n^{({\rm D})}} \int_0^{+\infty} {Q_1\left(\sqrt{\frac{\alpha_n^{({\rm SI})}}{\beta_n^{({\rm SI})}}},\sqrt{\frac{rz}{\beta_n^{({\rm SI})}}}\right)} \notag\\ 
& \quad\times \exp{\left(-\frac{\alpha_n^{({\rm D})}+z}{2\beta_n^{({\rm D})}}\right)I_0\left(\frac{\sqrt{\alpha_n^{({\rm D})}z}}{\beta_n^{({\rm D})}}\right)}dzdr.\label{Eq:conExp_hatR}
\end{align}
Thus, the expectation of $\hat{R}$ can be obtained by the expectation of \eqref{Eq:conExp_hatR} over $(\hat{\boldsymbol{a}}^{({\rm D})},\hat{\boldsymbol{a}}^{({\rm SI})})$, leading to the result in \eqref{Eq:Exp_fsapprox}.

\bibliographystyle{IEEEtran}

\begin{thebibliography}{10}
\providecommand{\url}[1]{#1}
\csname url@samestyle\endcsname
\providecommand{\newblock}{\relax}
\providecommand{\bibinfo}[2]{#2}
\providecommand{\BIBentrySTDinterwordspacing}{\spaceskip=0pt\relax}
\providecommand{\BIBentryALTinterwordstretchfactor}{4}
\providecommand{\BIBentryALTinterwordspacing}{\spaceskip=\fontdimen2\font plus
\BIBentryALTinterwordstretchfactor\fontdimen3\font minus
  \fontdimen4\font\relax}
\providecommand{\BIBforeignlanguage}[2]{{%
\expandafter\ifx\csname l@#1\endcsname\relax
\typeout{** WARNING: IEEEtran.bst: No hyphenation pattern has been}%
\typeout{** loaded for the language `#1'. Using the pattern for}%
\typeout{** the default language instead.}%
\else
\language=\csname l@#1\endcsname
\fi
#2}}
\providecommand{\BIBdecl}{\relax}
\BIBdecl
\bibitem{Ericsson6G}
Ericsson. ``Co-creating a cyber-physical world,'' Available [Online]: \url{https://www.ericsson.com/en/reports-and-papers/white-papers/co-creating-a-cyber-physical-world}, Last Accessed on 2024-08-12.
\bibitem{wang2023ontheroad}
C.-X. Wang {\em et al.}, ``On the road to {6G}: visions, requirements, key technologies, and testbeds,'' \emph{IEEE Commun. Surveys Tuts.}, vol.~25, no.~2, pp.~905--974, 2nd Quart. 2023.
\bibitem{Tariq-2020}
F. Tariq {\em et al.}, ``A speculative study on 6G,'' {\em IEEE Wireless Commun.}, vol. 27, no. 4, pp. 118--125, Aug. 2020.
\bibitem{xu2024intelligent}
S. Xu \emph{et al.}, ``Intelligent reflecting surface enabled integrated sensing, communication and computation,'' \emph{IEEE Trans. Wireless Commun.}, vol. 23, no. 3, pp. 2212--2225, Mar. 2024.


\bibitem{Sabharwal-2010}
M. Duarte and A. Sabharwal, ``Full-duplex wireless communications using off-the-shelf radios: Feasibility and first results,'' in {\em Proc. Asilomar Conf. Sig., Syst. \& Comput.}, Nov. 2010, pp. 1558--1562.
\bibitem{sabharwal2014inband}
A. Sabharwal {\em et al.}, ``In-band full-duplex wireless: challenges and opportunities,'' \emph{IEEE J. Sel. Areas in Commun.}, vol.~32, no.~9, pp.~1637--1652, Sept. 2014.
\bibitem{Kim2015asurvey}
B. Smida {\em et al.},
``Full-duplex wireless for {6G}: Progress brings new opportunities and challenges," {\em IEEE J. Sel. Areas in Commun.}, vol. 41, no. 9, pp. 2729--2750, Sep. 2023.
\bibitem{Kim2024asurvey}
Y. Kim, {\em et al.}, ``A state-of-the art survey on full duplex network design," \emph{Proc. of the IEEE}, vol. 112, no. 5, pp. 463-486, May 2024.
\bibitem{mohammadi2023acomprehensive}
M. Mohammadi, Z. Mobini, D. Galappaththige and C. Tellambura, ``A comprehensive survey on full-duplex communication: current solutions, future trends, and open issues,'' \emph{IEEE Commun. Surveys Tuts.}, vol.~25, no.~4, pp.~2190--2244, 4th Quart. 2023.

\bibitem{suk2022full}
G. Y. Suk {\em et al.}, ``Full duplex integrated access and backhaul for 5G NR: analyses and prototype measurements,'' {\em IEEE Wireless Commun.}, vol.~29, no.~4, pp.~40--46, Aug. 2022.
\bibitem{chen20235Gadvanced}
W. Chen {\em et al.}, ``5G-advanced toward 6G: Past, present, and future,'' \emph{IEEE J. Sel. Areas Commun.}, vol.~41, no.~6, pp.~1592--1619, Jun. 2023.
\bibitem{Abdelghaffar2024subband}
M. Abdelghaffar, T. V. P. Santhappan, Y. Tokgoz, K. Mukkavilli and T. Ji, ``Subband full-duplex large-scale deployed network designs and tradeoffs,'' {\em Proc. IEEE}, vol.~112, no.~5, pp.~487--510, May 2024.
\bibitem{everett2014passive}
E. Everett, A. Sahai, and A. Sabharwal, ``Passive self-interference suppression for full-duplex infrastructure nodes,'' \emph{IEEE Trans. Wireless Commun.}, vol.~13, no.~2, pp.~680--694, Jan. 2014.
\bibitem{bha2013full}
D. Bharadia, E. McMilin, and S. Katti, ``Full duplex radios,'' in \emph{Proc. ACM SIGCOMM}, pp.~375--386, Aug. 2013.
\bibitem{mas2017channel}
A. Masmoudi and T. Le-Ngoc, ``Channel estimation and self-interference cancelation in full-duplex communication systems,'' \emph{IEEE Trans. Veh. Technol.}, vol.~66, no.~1, pp.~321--334, Jan. 2017.
\bibitem{le2022atwo}
A. T. Le, X. Huang and Y. J. Guo, ``A two-stage analog self-interference cancelation structure for high transmit power in-band full-duplex radios,'' \emph{IEEE Wireless Commun. Lett.}, vol.~11, no.~11, pp.~2425--2429, Nov. 2022.
\bibitem{Hong2023frequency}
Z. H. Hong {\em et al.}, ``Frequency-domain RF self-interference cancellation for in-band full-duplex communications,'' {\em IEEE Trans. Wireless Commun.}, vol. 22, no. 4, pp. 2352--2363, April 2023.
\bibitem{muranov2021ondeep}
K. Muranov, M. A. Islam, B. Smida, and N. Devroye, ``On deep learning assisted self-interference estimation in a full-duplex relay link,'' \emph{IEEE Wireless Commun. Lett.}, vol.~10, no.~12, pp.~2762--2766, Dec. 2021.
\bibitem{dong2024augmentation}
Q. Dong, A. C. M. Austin and K. W. Sowerby, ``Augmentation of self-interference cancellation for full-duplex using NARX neural networks,'' \emph{IEEE Wireless Commun. Lett.}, vol.~13, no.~3, pp.~810--813, Mar. 2024.
\bibitem{elsayed2025ahybrid}
M. Elsayed and O. A. Dobre, ``A hybrid quantum-classical machine learning approach for self-interference cancellation in full-duplex transceivers,'' \emph{IEEE Commun. Lett.}, vol.~29, no.~4, pp.~774--778, Apr. 2025.

\bibitem{wong2021FAS}
K. K. Wong, A. Shojaeifard, K. F. Tong, and Y. Zhang, ``Fluid antenna systems,'' {\em IEEE Trans. Wireless Commun.}, vol. 20, no. 3, pp. 1950--1962, Mar. 2021.
\bibitem{wong2020FAS}
K. K. Wong, K. F. Tong, Y. Zhang, and Z. Zheng, ``Fluid antenna system for {6G}: When {Bruce Lee} inspires wireless communications,'' {\em Elect. Lett.}, vol.~56, no.~24, pp.~1288--1290, Nov. 2020.
\bibitem{New2024aTutorial}
W. K. New, {\em et al.}, ``A tutorial on fluid antenna system for 6G networks: Encompassing communication theory, optimization methods and hardware designs,'' \emph{IEEE Commun. Surv. \& Tutor.}, early access, \url{doi: 10.1109/COMST.2024.3498855}, Nov. 2024.
\bibitem{Lu-2025}
W.-J. Lu {\em et al.}, ``Fluid antennas: Reshaping intrinsic properties for flexible radiation characteristics in intelligent wireless networks,'' {\em IEEE Commun. Mag.}, vol. 63, no. 5, pp. 40--45, May 2025.

\bibitem{Hoang-2021}
T. V. Hoang, V. Fusco, T. Fromenteze and O. Yurduseven, ``Computational polarimetric imaging using two-dimensional dynamic metasurface apertures,'' {\em IEEE Open J. Antennas \& Propag.}, vol. 2, pp. 488--497, 2021.
\bibitem{Deng-2023}
R. Deng {\em et al.}, ``Reconfigurable holographic surfaces for ultra-massive MIMO in 6G: Practical design, optimization and implementation,'' {\em IEEE J. Select. Areas Commun.}, vol. 41, no. 8, pp. 2367--2379, Aug. 2023.
\bibitem{Shen-tap_submit2024}
Y. Shen {\em et al.}, ``Design and implementation of mmWave surface wave enabled fluid antennas and experimental results for fluid antenna multiple access,'' {\em arXiv preprint}, \url{arXiv:2405.09663}, May 2024.
\bibitem{zhang2024pixel}
J.~Zhang {\em et al.}, ``A novel pixel-based reconfigurable antenna applied in fluid antenna systems with high switching speed,'' \emph{IEEE Open J. Antennas \& Propag.}, vol. 6, no. 1, pp. 212--228, Feb. 2025.
\bibitem{Liu-2025arxiv}
B. Liu, K. F. Tong, K. K. Wong, C.-B. Chae, and H. Wong, ``Be water, my antennas: Riding on radio wave fluctuation in nature for spatial multiplexing using programmable meta-fluid antenna,'' {\em arXiv preprint}, \url{	arXiv:2502.04693}, 2025.

\bibitem{Khammassi2023}
M. Khammassi, A. Kammoun and M.-S. Alouini, ``A new analytical approximation of the fluid antenna system channel,'' {\em IEEE Trans. Wireless Commun.}, vol. 22, no. 12, pp. 8843--8858, Dec. 2023.
\bibitem{new2024fluid}
W. K. New, K. K. Wong, H. Xu, K. F. Tong and C.-B. Chae, ``Fluid antenna system: New insights on outage probability and diversity gain,''  {\em IEEE Trans. Wireless Commun.}, vol. 23, no. 1, pp. 128--140, Jan. 2024.
\bibitem{new2023information}
W. K. New, K. K. Wong, H. Xu, K. F. Tong and C.-B. Chae, ``An information-theoretic characterization of MIMO-FAS: Optimization, diversity-multiplexing tradeoff and $q$-outage capacity,'' {\em IEEE Trans. Wireless Commun.}, vol. 23, no. 6, pp. 5541--5556, Jun. 2024.

\bibitem{xu2023channel}
H. Xu {\em et al.}, ``Channel estimation for {FAS}-assisted multiuser {mmWave} systems,'' {\em IEEE Commun. Lett.}, vol.~28, no.~3, pp.~632--636, Mar. 2024.
\bibitem{new2025channel}
W. K. New {\em et al.}, ``Channel estimation and reconstruction in fluid antenna system: Oversampling is essential,'' {\em IEEE Trans. Wireless Commun.}, vol. 24, no. 1, pp. 309--322, Jan. 2025.
\bibitem{zhang2025successive}
Z. Zhang, J. Zhu, L. Dai and R. W. Heath, Jr., ``Successive Bayesian reconstructor for channel estimation in fluid antenna systems,'' {\em IEEE Trans. Wireless Commun.}, vol. 24, no. 3, pp. 1992--2006, March 2025.
\bibitem{hong2024coded}
H. Hong, K. K. Wong, K. F. Tong, H. Shin, and Y. Zhang, ``Coded fluid antenna multiple access over fast fading channels,'' \emph{IEEE Wireless Commun. Lett.}, vol.~14, no.~4, pp.~1249--1253, Apr. 2025.
\bibitem{hong2025Downlink}
H. Hong {\em et al.}, ``Downlink OFDM-FAMA in 5G-NR systems,'' {\em arXiv preprint}, \url{arxiv:2501.06974}, Jan. 2025.
\bibitem{hong2025fluid}
H. Hong {\em et al.}, ``Fluid antenna system empowering 5G NR,'' {\em arXiv preprint}, \url{arxiv:2503.05384}, Mar. 2025.

\bibitem{wong2022FAMA}
K. K. Wong and K. F. Tong, ``Fluid antenna multiple access,'' \emph{IEEE Trans. Wireless Commun.}, vol.~21, no.~7, pp. 4801--4815, Jul. 2022.
\bibitem{Wong2024cuma}
K. K. Wong, C. B. Chae, and K. F. Tong, ``Compact ultra massive antenna array: A simple open-loop massive connectivity scheme,'' {\em IEEE Trans. Wireless Commun.}, vol. 23, no. 6, pp. 6279--6294, Jun. 2024.
\bibitem{Xu2024revisiting}
H.~Xu {\em et al.}, ``Revisiting outage probability analysis for two-user fluid antenna multiple access system,'' \emph{IEEE Trans. Wireless Commun.}, vol. 23, no. 8, pp. 9534--9548, Aug. 2024.
\bibitem{xu2023capacity}
H. Xu {\em et al.}, ``Capacity maximization for FAS-assisted multiple access channels,'' {\em IEEE Trans. Commun.}, early access, \url{doi:10.1109/TCOMM.2024.3516499}, Dec. 2024.

\bibitem{ghadi2024on}
F. Rostami Ghadi {\em et al.}, ``On performance of RIS-aided fluid antenna systems,'' {\em IEEE Wireless Commun. Lett.}, vol. 13, no. 8, pp. 2175--2179, Aug. 2024.
\bibitem{salem2025first}
A. Salem {\em et al.}, ``A first look at the performance enhancement potential of fluid reconfigurable intelligent surface,'' {\em arXiv preprint}, \url{arXiv:2502.17116v1}, 2025.
\bibitem{xiao2025fluid}
H. Xiao {\em et al.}, ``Fluid reconfigurable intelligent surfaces: Joint on-off selection and beamforming with discrete phase shifts,'' {\em arXiv preprint}, \url{arXiv:2503.14601v1}, 2025.
\bibitem{ghadi2025fires}
F. Rostami Ghadi {\em et al.}, ``FIRES: Fluid integrated reflecting and emitting surfaces,'' {\em arXiv preprint}, \url{arXiv:6455283}, 2025.

\bibitem{sko2023full}
C. Skouroumounis and I. Krikidis, "Fluid antenna-aided full duplex communications: a macroscopic point-of-view," \emph{IEEE J. Sel. Areas Commun.}, vol.~41, no.~9, pp.~2879-2892, Sept. 2023.
\bibitem{buzzi2016on}
S. Buzzi and C. D'Andrea, ``On clustered statistical MIMO millimeter wave channel simulation,'' {\em arXiv preprint}, \url{arXiv:1604.00648v2}, May 2016.
\end{thebibliography}

\end{document}